\documentclass[12pt,draftcls,onecolumn]{IEEEtran}
\usepackage{cite}
\usepackage{graphicx}
\usepackage{amsmath,amssymb,amsfonts}

\makeatletter
\if@twocolumn%
\newcommand{\Figwidth}{\columnwidth}%
\def\twocolbreak{\nonumber\\ &}%
\def\twocolnewline{\nonumber\\}%
\def\twocolnewlineonly{\\}%
\def\twocolAlignMarker{&}%
\else
\newcommand{\Figwidth}{4.5in}%
\def\twocolbreak{}%
\def\twocolnewline{}%
\def\twocolnewlineonly{}%
\def\twocolAlignMarker{}%
\fi%
\makeatother%

\def\hfeat{h_{m_{k}}^{k}}

\allowdisplaybreaks

\begin{document}
\title{Community Detection with Side Information: Exact Recovery under the Stochastic Block Model}

\author{Hussein Saad, {\em Student Member, IEEE} and Aria Nosratinia, {\em Fellow, IEEE}
\thanks{This work was supported in part by the NSF grant 1711689.}
\thanks{The authors are with the Department of Electrical Engineering, University of Texas at
    Dallas, Richardson, TX 75083-0688 USA, E-mail:
    hussein.saad@utdallas.edu; aria@utdallas.edu.}
\thanks{Some of the results of this paper appeared in the Allerton Conference on Communications, Control, and Computing 2017.}
}
\maketitle

\newtheorem{theorem}{Theorem}
\newtheorem{lemma}{Lemma}
\newtheorem{remark}{Remark}
\newtheorem{definition}{Definition}

\def\A{A}
\def\B{\mathbf B}
\def\D{\mathbf D}
\def\Dset{\mathbb D}
\def\E{E}
\def\d{d}
\def\G{\mathbf G}
\def\g{\mathbf g}
\def\Gset{\mathcal G}
\def\H{\mathbf H}
\def\h{\mathbf h}
\def\Hset{\mathbb H}
\def\I{\mathbf I}
\def\J{J}
\def\Jset{\mathbb J}
\def\K{K}
\def\l{\ell}
\def\N{N}
\def\M{M}
\def\P{\mathbf P}
\def\Q{Q}
\def\Qset{\mathbb Q}
\def\R{R}
\def\q{q}
\def\r{r}
\def\S{\mathcal{S}}
\def\T{T}
\def\u{\mathbf u}
\def\U{\mathbf U}
\def\Uset{\mathbb U}
\def\v{\mathbf v}
\def\V{\mathbf V}
\def\W{\mathbf W}
\def\w{\mathbf w}
\def\X{\mathbf X}
\def\x{\mathbf x}
\def\Y{\mathbf Y}
\def\y{\mathbf y}
\def\Z{\mathbf Z}
\def\z{\mathbf z}
\def\bigO{O}
\def\littleO{o}
\def\ZeroMat{\mathbf{0}}


\begin{abstract}
The community detection problem involves making inferences about node labels in a graph, based on observing the graph edges. This paper studies the effect of additional, non-graphical side information on the phase transition of exact recovery in the binary stochastic block model (SBM) with $n$ nodes. When side information consists of noisy labels with error probability $\alpha$, it is shown that phase transition is improved if and only if $\log(\frac{1-\alpha}{\alpha})=\Omega(\log(n))$. When side information consists of revealing a fraction $1-\epsilon$ of the labels, it is shown that phase transition is improved if and only if $\log(1/\epsilon)=\Omega(\log(n))$. For a more general side information consisting of $K$ features, two scenarios are studied: (1)~$K$ is fixed while the likelihood of each feature with respect to corresponding node label evolves with $n$, and (2)~The number of features $K$ varies with $n$ but the likelihood of each feature is fixed. In each case, we find when side information improves the exact recovery phase transition and by how much. 
In the process of deriving inner bounds, a variation of an efficient algorithm is proposed for community detection with side information that uses a partial recovery algorithm combined with a local improvement procedure. 
\end{abstract}


\IEEEpeerreviewmaketitle

\section{Introduction}

The problem of learning or detecting community structures in random graphs has been studied in statistics~\cite{SBM_first,Min_max_SBM,stat1,stat2,stat3}, computer science~\cite{SCT,CS1,CS2,CS3,CS4} and theoretical statistical physics~\cite{Phys,Phys2}. Detection of communities on graphs is motivated by applications including finding like-minded people in social networks~\cite{social}, improving recommendation systems~\cite{recommendation}, and detecting protein complexes~\cite{protein}. Among the different random graph models~\cite{survey,survey1}, the stochastic block model (SBM) is widely used in the context of community detection\cite{exact_general_sbm}. This extension of the Erd\"{o}s-Renyi model consists of $n$ nodes that belong to two communities, each pair of nodes connected with probability $p$ if the pair belongs to the same community, and with probability $q$ otherwise. The prior distribution of the node labels is identical and independent, and often uniform (labels are equi-probable). The goal of community detection is to recover/detect the labels upon observing the graph edges. 

Random graphs experience measure concentration in the recovery of labels~\cite{exact_general_sbm}, i.e., for some underlying graph distributions, recovered labels will become reliable as the size of data set increases, and for others they do not. The boundary of this phenomenon is often described as a phase transition~\cite{exact_general_sbm}. The location of this phase transition and the set of graphs that fall inside the region described by it, is an important indicator of the broad class of graph-based problems that are reliably solvable in the context of community detection. Much of the theoretical work on community detection~\cite{weak_sbm1,weak_sbm,partial_alg,weak_sbm2,correlated,partial_sbm1,density,our,exact_sbm,Consistency,exact_general_sbm} concentrates on characterizing this phase transition and understanding its properties.

The literature on community detection has, for the most part, concentrated on purely graphical observations. However, in many practical applications, non-graphical relevant information is available that can aid the inference. For example, social networks such as Facebook and Twitter have access to much information other than the graph edges. A citation network has the authors' names, keywords, and abstracts of papers, and therefore may provide significant additional information beyond the co-authoring relationships. Figure~\ref{example2} illustrates standard community detection as well as community detection with side  information.  This paper presents new results on the utility of side information in community detection, in particular shedding light on the conditions under which side information can improve the phase transition of community detection, and the magnitude of the improvement.

\begin{figure}
\centering
\begin{minipage}{0.38\columnwidth}
\centering
\includegraphics[width=1.25in]{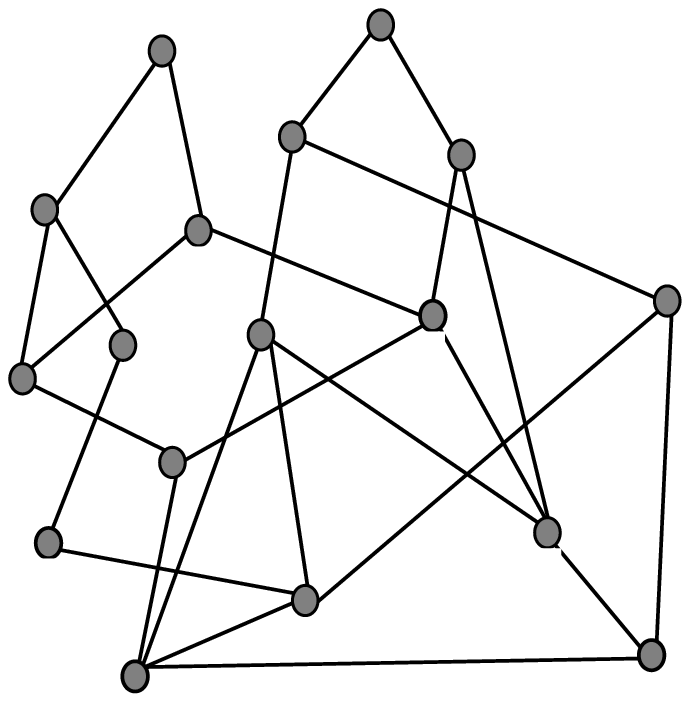}\\
\footnotesize{\em Observed Graph} 
\end{minipage}
 {$\displaystyle \xrightarrow[\text{detection}]{\text{community}}$}
\begin{minipage}{0.38\columnwidth}
\centering 
\includegraphics[width=1.4in]{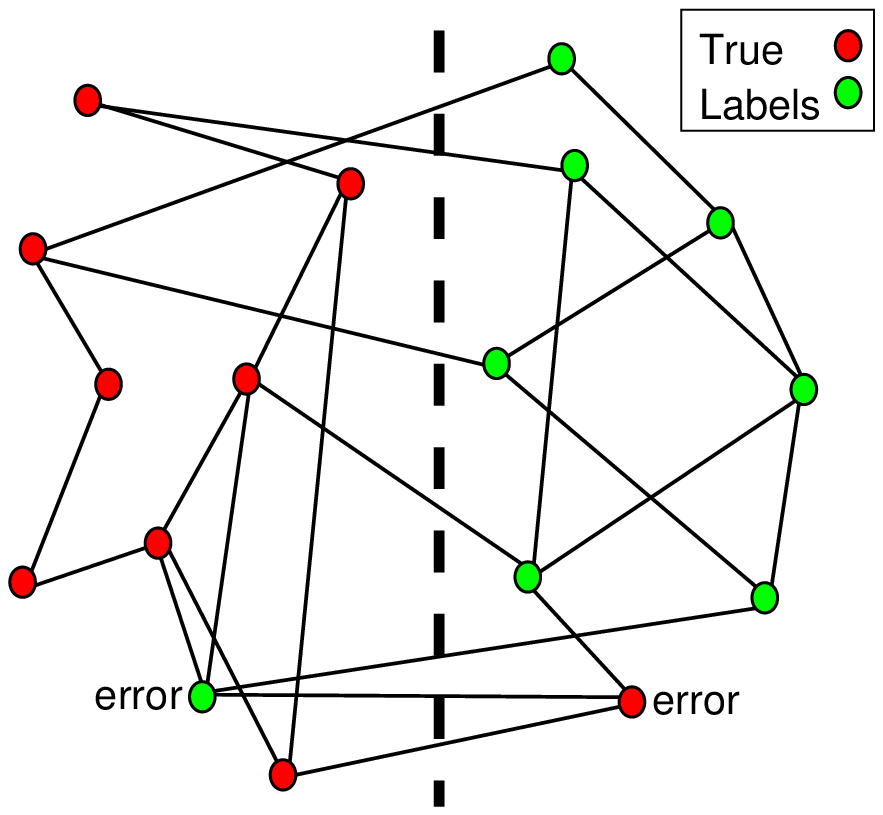}\\
\footnotesize{\em Detected communities}
\end{minipage}
\\[0.2in]
\centering
\begin{minipage}{0.38\columnwidth}
\centering
\includegraphics[width=1.25in]{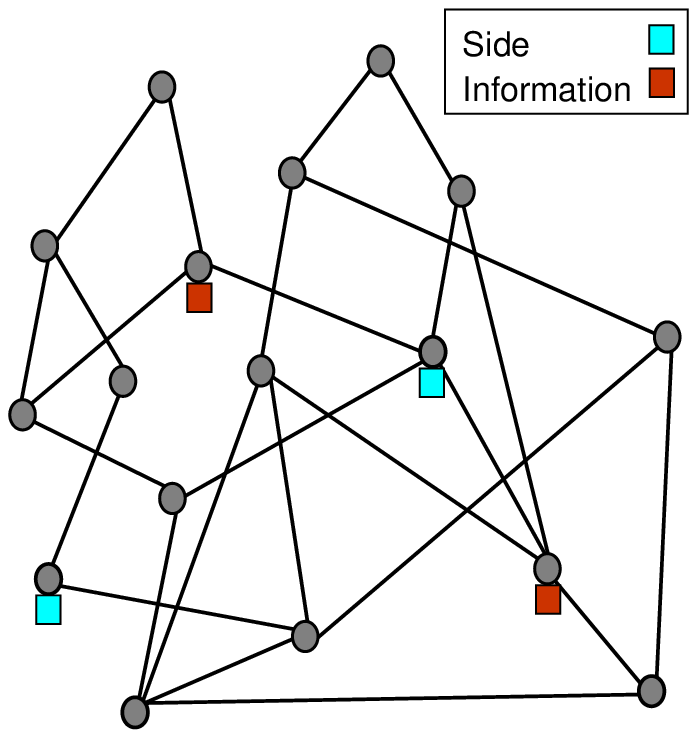}\\
\footnotesize{\em Graph + side information}
\end{minipage}
{$\displaystyle \xrightarrow[\text{detection}]{\text{community}}$}
\begin{minipage}{0.38\columnwidth}  
\centering 
\includegraphics[width=1.4in]{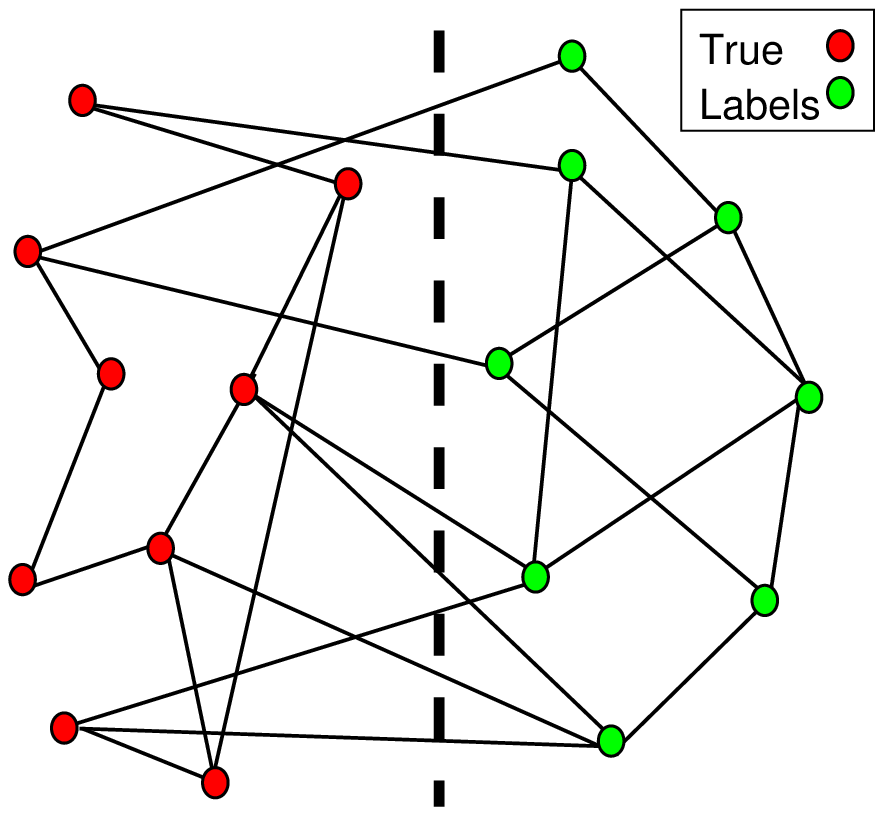}\\
\footnotesize{\em Enhanced detection}
\end{minipage}
\caption{(top) standard community detection (bottom) Community detection with side information} 
\label{example2}
\end{figure}

Community detection outcomes fall into several broad categories in terms of residual error as the size of the graph $n$ grows, enumerated here in increasing order of strength: Correlated recovery refers to community detection that performs better than random guessing~\cite{weak_sbm1,weak_sbm,partial_alg,weak_sbm2,correlated}. Weak recovery means the fraction of misclassified labels in the graph vanishes with probability converging to one~\cite{partial_sbm1,density,our}. Exact recovery means correct recovery of all nodes with probability converging to one~\cite{exact_sbm,Consistency,exact_general_sbm}. This paper concentrates on the exact recovery metric.\footnote{Formally, let $e_n$ denote the number of misclassified nodes. Then, correlated recovery means  $\lim_{n\rightarrow\infty}\mathbb{P}(\frac{e_n}{n} < 0.5)=1$. Weak recovery means $\lim_{n\rightarrow\infty}\mathbb{P}(\frac{e_n}{n} <\varepsilon)=1$ for any positive $\varepsilon$. Exact recovery means $\lim_{n \rightarrow\infty}\mathbb{P}(e_n= 0) =1$.}

A few results have recently appeared in the literature on the broader community detection problem in the presence of additional (non-graphical) information. Mossel and Xu~\cite{MoselXu:ACM16} studied the behavior of belief propagation detector in the presence of noisy label information. Cai {\em et al.}~\cite{Cai:BP-BEC} studied the effect of knowing a growing fraction of labels on correlated and weak recovery. Neither of~\cite{MoselXu:ACM16,Cai:BP-BEC} includes a converse, so they do not establish phase transition. Kadavankandy {\em et al.}~\cite{Kadavankandy:SingleCommunity} studied the single-community problem with noisy label observations, showing weak recovery in the sparse regime. Kanade {\em et al.}~\cite{Kanade:GlobalLocal} showed that partial observation of labels is unhelpful to the correlated recovery phase transition if a vanishing portion of labels are available. The exact recovery metric is not addressed in these works, and they do not establish a phase transition under side information.\footnote{Arguably the closest result in the literature to our work can be found in~\protect\cite[Theorem 4]{Abbe_1}, which is discussed in Section~\ref{sec:varying}.}

In the interest of completeness, we also mention the following works even though they have a very different perspective. In statistics, several works have appeared on model-matching to real data consisting of both graphical and non-graphical observations, where additional information such as ``annotation''~\cite{Newman:annotated}, ``attributes''~\cite{Yang:attribute}, or ``features''~\cite{Zhang:features} has been considered. These works aim at model matching to real (finite) data sets, and propose a parametric model that expresses the joint probability distribution of the graphical and non-graphical (attribute/feature) observations. Although the focus of these papers is very different from the present paper, they nevertheless show the interest of the broader community in modeling side-information for graph-based inference.

The following observations further motivate this work. For the exact recovery metric, the effect of side information has not been comprehensively studied. Even for correlated recovery and weak recovery, the effect of side information has only been studied for belief propagation, which is not enough to establish phase transition. In the context of binary labels, only binary side information (possibly with erasures) has been studied. Practical scenarios motivate the study of more general side information whose alphabet does not match the number/identity of communities. Also of interest is side information consisting of several (potentially non-binary) features, which has not been thoroughly investigated either in the context of belief propagation or maximum likelihood, although~\protect\cite[Theorem 4]{Abbe_1} opened the subject in a special setting. 

\section{System Model and Contributions}\label{sys.}

We consider the binary symmetric stochastic block model, with community labels denoted $1$ and $-1$. The number of nodes in the graph is denoted with $n$. The node labels are independent and identically distributed across $n$, with $1$ and $-1$ labels having equal probability. If two nodes belong to the same community, there is an edge between them with probability $p=a\frac{\log(n)}{n}$, and if they are from different communities, there is an edge between them with probability $q=b\frac{\log(n)}{n}$. Finally, for each node one or more scalar random variables are observed containing side information. Conditioned on node labels, the side information of different nodes are assumed to be independent of each other and of the graph edges. Three models for this side information are considered.

In the first model, for each node, a scalar side information is observed which is the true label with probability $(1-\alpha)$ and its complement (false) with probability $\alpha$, where $\alpha \in (0,0.5)$. In the second model, for each node, a scalar side information is observed which is the true label with probability $1-\epsilon$ or $0$ (erased) with probability $\epsilon$, where $\epsilon \in (0,1)$. In the third model, we consider side information consisting of $K$ random variables (features) with finite cardinalities $M_{k}$, $k \in \{1,\cdots,K\}$.

The observed graph is denoted by $G$, the vector of nodes' true assignment by $\boldsymbol{x}^{*}$, and the nodes' side information by vector $\boldsymbol{y}$ when each node has a scalar side information, or with collection of length-$n$ vectors $\boldsymbol{y_{k}}, \; k=1,\ldots,K$ when side information for each node consists of $K$ features. The goal is to recover the node assignment $\boldsymbol{x}^{*}$ from the observation of the graph $G$ and side information.

In this paper, exact recovery is considered in the dense regime, i.e., when $p = a \frac{\log n}{n}$ and $q = b \frac{\log n}{n}$ with constants $a \geq b > 0$. In this regime the exact recovery phase transition without side information is $(\sqrt{a} - \sqrt{b})^{2}> 2$~\cite{exact_sbm}. We investigate the question: when and by how much can side information affect the phase transition threshold of exact recovery?  The contributions of this paper are as follows:

\begin{itemize}

\item When side information consists of observing node labels with erasure probability $\epsilon \in (0,1)$, we show that if $\log(\epsilon)=o(\log(n))$, the phase transition is not improved by side information. On the other hand, if $\log(\epsilon) = -\beta\log(n) +o(\log(n))$ for some $\beta>0$, i.e., $O(\log(n))$, a necessary and sufficient condition for exact recovery is $(\sqrt{a} - \sqrt{b})^{2} + 2\beta> 2$. 
  
\item When side information consists of observing node labels with error probability $\alpha \in (0,0.5)$, if $c = \log(\frac{1-\alpha}{\alpha})$ is $o(\log(n))$, then the phase transition is not improved by side information. On the other hand, if $c = \beta\log(n) +o(\log(n)), \beta>0$, i.e., $O(\log(n))$, necessary and sufficient conditions for exact recovery are derived as follows:
  \[
  \begin{cases}
   \eta(a,b,\beta) > 2 &\text{when }   \beta < \frac{T(a-b)}{2} \\
\beta > 1 & \text{when } \beta > \frac{T(a-b)}{2}
    \end{cases}
  \]
  with the following parameters defined for convenience:
\begin{align}
  \eta(a,b,\beta) &\triangleq a+b+\beta-\frac{2\gamma}{T}+\frac{\beta}{T}\log(\frac{\gamma + \beta}{\gamma - \beta}) \label{eta}\\
  T&\triangleq\log(\frac{a}{b}), \quad \gamma \triangleq \sqrt{\beta^{2} + abT^{2}} \label{gamma}
\end{align}
An early version of this result appeared in~\protect\cite{our2}.

\item When side information consists of $K$ features each with finite and fixed cardinality, two scenarios are considered: (1) $K$ is fixed while the conditional distribution of each feature varies with $n$. In this scenario, we study how the quality of each feature must evolve as the size of the graph grows, so that phase transition can be improved. (2) $K$ varies with $n$ while the conditional distribution of features is fixed. In this scenario, the quality of the features is independent of $n$, and we study how many features are needed in addition to the graphical information, so that the phase transition can be improved.

\item
Sufficient conditions are provided via an efficient algorithm employing partial recovery and a local improvement using both the graph and the side information. The two-step recovery algorithm without side information appeared in~\cite{exact_sbm,exact_general_sbm,NIPS2016_6196}. In this paper, it is refined and generalized in the presence of side information.
\end{itemize}

\begin{remark}
In earlier community detection problems~\cite{exact_sbm,exact_general_sbm}, LLRs do not depend on $n$ even though individual likelihoods (obviously) do. This was very fortunate for calculating asymptotics. In the presence of side information, this convenience disappears and LLRs will now depend on $n$, creating complications in bounding error event probabilities en route to finding the threshold in the asymptote of large $n$. Overcoming this technical difficulty is part of the contributions of this paper.
\end{remark}

To illustrate the results of this paper, Figures~\ref{fig.1},~\ref{fig.2} show the error exponent for the side information consisting of partially revealed labels or noisy label observation, as a function of $\beta$. It is observed that the value of $\beta$ needed for recovery depends on $a, b$. For the partially revealed labels, when $(\sqrt{a} - \sqrt{b})^{2} < 2$, the critical $\beta$ is $1-\frac{1}{2}(\sqrt{a} - \sqrt{b})^{2}$. For noisy label observations, when  $(\sqrt{a} - \sqrt{b})^{2} < 2$, the value of critical $\beta$ can be determined as follows: if $\eta(a,b,\frac{T(a-b)}{2}) > 2$, then the critical $\beta$ is the solution to $\eta = 2$. On the other hand, if $\eta(a,b,\frac{T(a-b)}{2}) < 2$, then the critical $\beta$ is one.

\begin{figure}[t]
\centering
\includegraphics[width=\Figwidth]{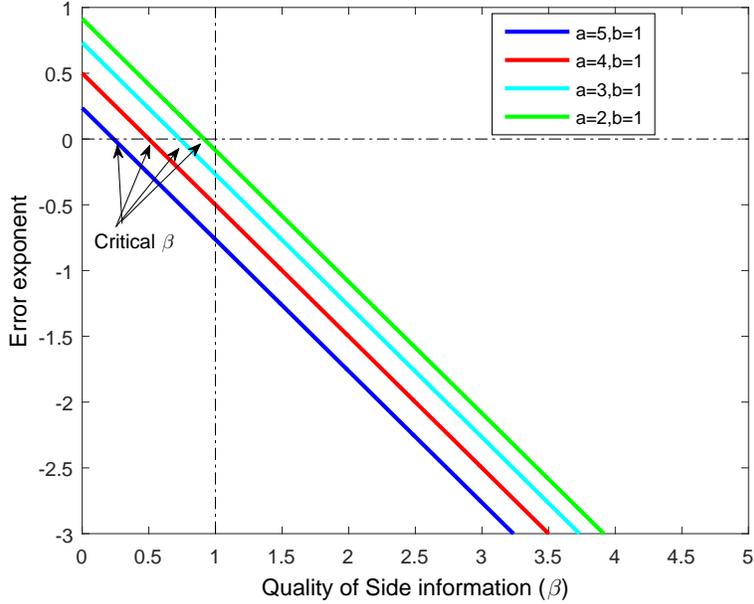}
\caption{Error exponent for noisy label observations as a function of $\beta$.}
\label{fig.1}
\end{figure}
\begin{figure}
\centering
\includegraphics[width=\Figwidth]{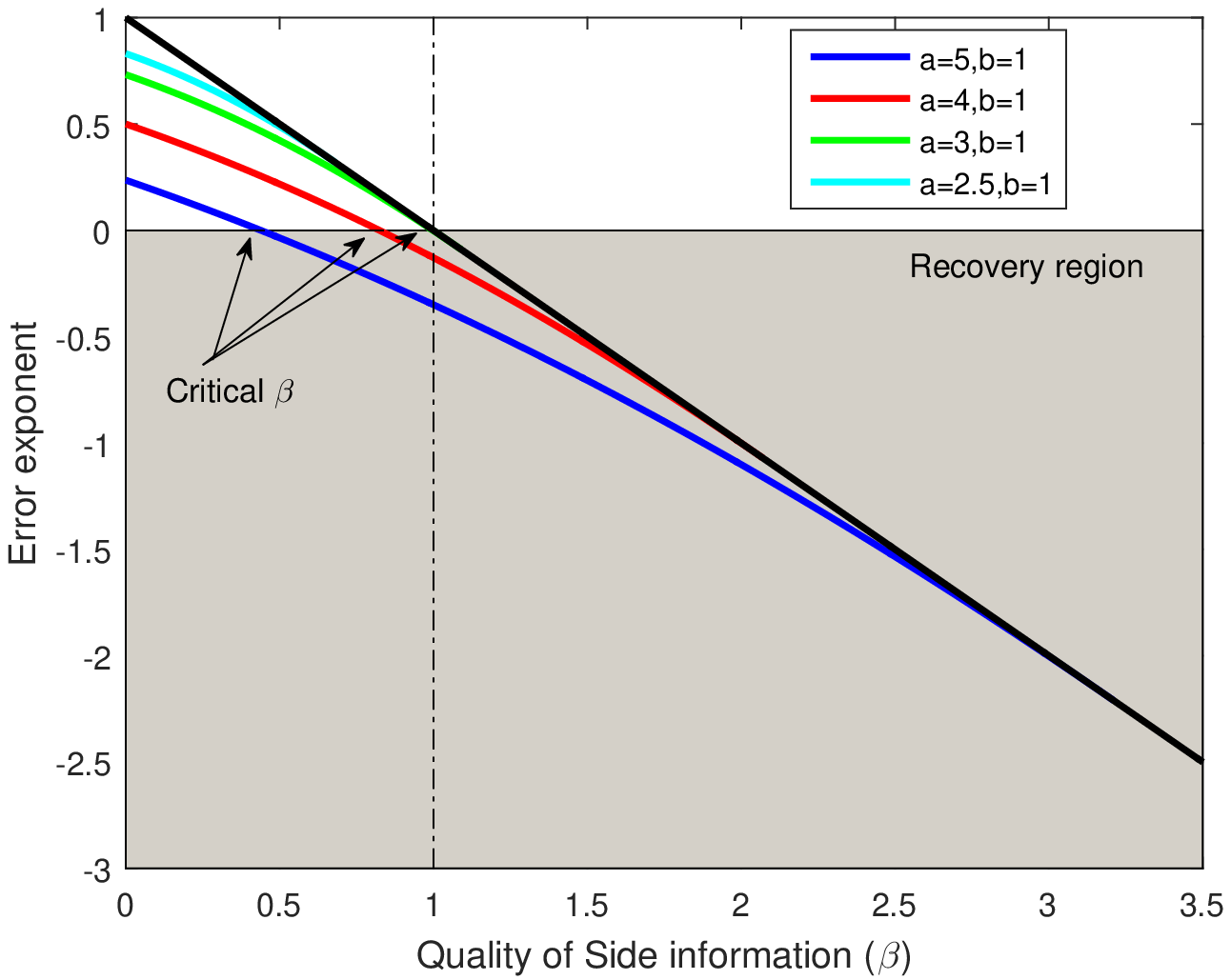}
\caption{Error exponent of partial label observation as a function of $\beta$.}
\label{fig.2}
\end{figure}

\section{Noisy Label Side Information}\label{sec1}

In this section, side information consists of a noisy version of the label that with probability $\alpha \in (0,0.5)$ fails to match the true label. 

We begin by calculating the maximum likelihood rule for detecting the communities under side information. The maximum likelihood detector {\em without} side information~\cite{exact_sbm} is the minimizer of the number of edges between two detected communities, subject to both detected communities having size $\frac{n}{2}$. The set of nodes belonging to the two communities are denoted with $A$ and $B$, i.e., $A \triangleq \{i : x_i=1\}$ and $B \triangleq \{i : x_i=-1\}$. $E(A)$ denotes the number of edges whose two vertices belong to community $A$, and $E(B)$ the number of edges whose two vertices belong to community $B$. The total number of edges in the graph is denoted $E_{t}$. Also, define:
\begin{align*}
J_+(A) &\triangleq \big | \{ i \in A : y_i=1 \}\big |\\
J_-(B) &\triangleq \big | \{ i \in B : y_i=-1 \}\big |
\end{align*}
Then, the log-likelihood function can be written as:
\begin{align}\label{main.rule.Full.noisy}
& \log\big(\mathbb{P}(G,\boldsymbol{y}|\boldsymbol{x})\big) \overset{(a)}{=} \log\big(\mathbb{P}(G|\boldsymbol{x})\big) + \log\big(\mathbb{P}(\boldsymbol{y}|\boldsymbol{x})\big)  \nonumber\\
= & \log\big( p^{E(A)+E(B)} q^{E_{t}-E(A)-E(B)} (1-p)^{2 {{\frac{n}{2}}\choose{2}} - E(A)-E(B)} \nonumber\\
& (1-q)^{\frac{n^{2}}{4} - E_{t}+E(A)+E(B)}  \big) + \log\big( (1-\alpha)^{J_{+}(A)+J_{-}(B)} \twocolbreak
\alpha^{n-J_{+}(A)-J_{-}(B)}  \big) \nonumber \\ 
 \overset{(b)}{=} & R + T\big(E(A)+E(B)\big)(1+o(1)) + c\big(J_{+}(A)+J_{-}(B)\big)
\end{align}
where $(a)$ holds because $G,\boldsymbol{y}$ are independent given $\boldsymbol{x}$. In $(b)$, all terms that are independent of $\boldsymbol{x}$ have been collected into a constant $R$, and  $\log(\frac{p(1-q)}{q(1-p)})$ has been approximated by $(1+o(1))T$, which is made possible because $(1-p), (1-q)$ both approach $1$ as $n\to\infty$. The difference between Eq.~\eqref{main.rule.Full.noisy} and the likelihood function without side information is the term $c\big(J_{+}(A)+J_{-}(B)\big)$ and a constant $n\log \alpha$ that is hidden inside $R$.

The following lemma characterizes a lower bound on the probability of failure of the maximum likelihood detector. Let $E[\cdot,\cdot]$ denote the number of edges between two sets of nodes.\footnote{For economy of notation, in the arguments of $E[\cdot,\cdot]$ we represent singleton sets by their single member.}
\begin{lemma}
  \label{Le.1}
  Let $A$ and $B$ denote the true communities. Define the following events:
  \begin{align}\nonumber
F & \triangleq \{ \text{Maximum Likelihood Detector fails} \} \\ \nonumber
F_{A} & \triangleq \{ \exists i \in A: T(E[i,B] - E[i,A]) - cy_{i} \geq T \} \\ 
F_{B} & \triangleq \{ \exists j \in B: T(E[j,A] - E[j,B]) + cy_{j} \geq T \} 
\end{align}
Then, $F_A \cap F_B \Rightarrow F$.
\end{lemma}

\begin{proof}
Define two new communities $\hat{A} = A\backslash \{i\}\cup\{j\}$ and $\hat{B} = B\backslash \{j\}\cup\{i\}$. If $\log\big(\mathbb{P}(G,\boldsymbol{y}|\hat{A},\hat{B})\big) \geq \log\big(\mathbb{P}(G,\boldsymbol{y}|A,B)\big)$ it means maximum likelihood chooses incorrectly and therefore fails. We show that this happens under $F_A\cap F_B$.

Let $A_{ij}\sim Bern(q)$ be a random variable representing the existence of the edge between nodes $i$ and $j$. Then, using~\eqref{main.rule.Full.noisy}:
\begin{align}
\twocolAlignMarker \log\big(\mathbb{P}(G,\boldsymbol{y}|\hat{A},\hat{B})\big) \twocolnewline
&= R + T\big(E(\hat{A})+E(\hat{B})\big) + c\big(J_{+}(\hat{A})+J_{-}(\hat{B})\big) \nonumber \\ 
&= R + T\big(E(A)+E(B)\big) +c\big(J_{+}(A)+J_{-}(B)\big) -2TA_{ij} \nonumber \\ 
&  + T\big(E[j,A] - E[j,B] +  E[i,B]- E[i,A]\big) + c(y_{j}-y_{i}) \nonumber \\
&\overset{(a)}{\geq}  \log\big(\mathbb{P}(G,\boldsymbol{y}|A,B)\big) + 2T (1-A_{ij}) \nonumber\\
&\overset{(b)}{\geq} \log\big(\mathbb{P}(G,\boldsymbol{y}|A,B)\big)
\end{align}
where $(a)$ holds by the assumption that $F_{A} \cap F_{B}$ happened and $(b)$ holds because $(1-A_{ij}) \geq 0$ and $T \geq 0$. The inequality $(b)$ implies the failure of maximum likelihood.
\end{proof}


\subsection{Necessary Conditions}\label{LB.1}

\begin{theorem}\label{Th.1}
Define $c\triangleq \log(\frac{1-\alpha}{\alpha})$. The maximum likelihood failure probability is bounded away from zero if:
\[
  \begin{cases}
(\sqrt{a} - \sqrt{b})^{2}  < 2 \text{ when } c = o(\log(n)) \\
\eta(a,b,\beta) < 2 \text{ when }  c = (\beta+o(1))\log(n), 0<\beta <\frac{T(a-b)}{2} \\
   \beta<1 \qquad\quad \text{ when }  c = (\beta+o(1))\log(n), \beta >\frac{T(a-b)}{2} 
    \end{cases}
  \]
\end{theorem}

\begin{proof}
Since $\boldsymbol{x^{*}}$ is generated uniformly, the ML detector is optimal in error probability. Hence, if ML fails with nonzero probability, every other detector must fail with nonzero probability. So it suffices to establish the error probability of ML. The main difficulty in bounding the error probability of ML is the dependency between the graph edges. To overcome this dependency, we follow steps that are broadly similar to~\cite{exact_sbm}, but our bounding techniques involve Chernoff type arguments and Cramer and Sanov large deviation principles that are more compact than combinatorial techniques of~\cite{exact_sbm}. 
\begin{definition}
\label{Def.1}
Let $H$ be a subset of $A$ with $|H| = \frac{n}{\log^{3}(n)}$ and define the following events for each node $i\in H$:
\begin{align*}
\Delta_{i} =& \big\{E[i,H] \leq \frac{\log(n)}{\log\log(n)}\big\} \nonumber\\ 
F_{i}^{H} =& \bigg\{ TE[i,A\backslash H] + cy_{i} + T + T\frac{\log(n)}{\log\log(n)}  \leq TE[i,B] \bigg\} 
\end{align*}
and the following events defined on $H$:
\begin{align*}
\Delta =&  \cap_{i\in H} \Delta_i \nonumber\\ 
F^{H} =& \cup_{i \in H} F_{i}^{H}\nonumber
\end{align*}
\end{definition}
\begin{lemma}\label{Le.4}
If $\mathbb{P}(F^{H}) \ge 1-\delta$ and $\mathbb{P}(\Delta) \geq 1-\delta$ for $\delta<\frac{1}{4}$, then there exists a positive $\delta'$ so that  $\mathbb{P}(F) \geq \delta'$.
\end{lemma}
\begin{proof}
Clearly $\Delta \cap F^{H} \Rightarrow F_{A}$. Hence,
\begin{equation}\nonumber
\mathbb{P}(F_{A}) \geq \mathbb{P}(F^{H}) + \mathbb{P}(\Delta) - 1 \geq 1-2\delta
\end{equation}
By the symmetry of the graph and the side information, $\mathbb{P}(F_{B}) \geq 1-2\delta$ as well. Also, by Lemma~\ref{Le.1} $F_{A} \cap F_{B} \Rightarrow F$. Then:
\begin{equation}\nonumber
\mathbb{P}(F) \geq \mathbb{P}(F_{A}) + \mathbb{P}(F_{B}) - 1 \geq 1-4 \delta
\end{equation}
For $\delta< \frac{1}{4}$, $\mathbb{P}(F)$ is bounded away from zero.
\end{proof}

\begin{lemma}\label{Le.5}
$\lim_{n\rightarrow\infty}\mathbb{P}(\Delta) = 1$
\end{lemma}

\begin{proof}
Let $W_{i} \sim Bern(p)$. Then:
\begin{align}
\mathbb{P}(\Delta_{i}^{c}) & = \mathbb{P}\bigg(\sum_{j=1}^{i-1} W_{j} + \sum_{j=i+1}^{\frac{n}{\log^{3}(n)}} W_{j} \geq \frac{\log(n)}{\log(\log(n))}\bigg) \nonumber\\ 
& \leq  \mathbb{P}\bigg(\sum_{j=1}^{\frac{n}{\log^{3}(n)}} W_{j} \geq \frac{\log(n)}{\log(\log(n))}\bigg) \nonumber\\ 
& \leq \Big(\frac{1}{e} \frac{\log^{3}(n)}{a\log(\log(n))}\Big)^{\frac{-\log(n)}{\log(\log(n))}} \nonumber
\end{align}
via a multiplicative form of Chernoff bound, stating that a sequence of $n$ i.i.d random variables $X_{i}$, $\mathbb{P}(\sum_{i=1}^{n} X_{i} \geq t\mu) \leq (\frac{t}{e})^{-t\mu}$, where $\mu = n\mathbb{E}[X]$. Thus, by union bound:
\begin{align}
&\mathbb{P}(\Delta)  \geq 1- \frac{n}{\log^{3}(n)} \Big(\frac{1}{e} \frac{\log^{3}(n)}{a\log(\log(n))}\Big)^{\frac{-\log(n)}{\log(\log(n))}}  \nonumber\\
= & 1- e^{\log(n) - 3\log(\log(n))} \twocolbreak e^{\bigg[\frac{\log(n) \log(ae)}{\log(\log(n))} - \frac{\log(n)}{\log(\log(n))}\big(3\log(\log(n)) - \log(\log(\log(n))) \big) \bigg]}  \nonumber\\
= & 1- e^{-2\log(n) + o(\log(n))}  \nonumber
\end{align}
\end{proof}

\begin{lemma}\label{Le.6}
For any $\delta\in(0,1)$ and for sufficiently large $n$, if $\mathbb{P}(F_{i}^{H}) > \frac{\log^{3}(n)}{n} \log(\frac{1}{\delta})$, then $\mathbb{P}(F^{H})  \geq 1-\delta$.
\end{lemma}

\begin{proof}
Because $F_{i}^{H}$ are i.i.d.:
\begin{align}\nonumber
\mathbb{P}(F^{H}) & = \mathbb{P}(\cup_{i \in H }F_{i}^{H}) = 1 - \mathbb{P}(\cap_{i \in H} (F_{i}^{H})^{c}) \\ 
&= 1-\big[\big(1-\mathbb{P}(F_{i}^{H})\big)^{\frac{1}{\mathbb{P}(F_{i}^{H})}}\big]^{(\frac{n\mathbb{P}(F_{i}^{H})}{\log^{3}(n)})}  \label{Eq:NotO1}\\
&>1-\big[\big(1-\mathbb{P}(F_{i}^{H})\big)^{\frac{1}{\mathbb{P}(F_{i}^{H})}}\big]^{-\log\delta} \nonumber
\end{align}
where the last inequality holds by the statement of the Lemma. If $\mathbb{P}(F_{i}^{H})$ is $o(1)$, then the quantity inside the bracket tends to $e^{-1}$ and the result follows.  If $\mathbb{P}(F_{i}^{H})$ is not $o(1)$, then from Eq.~\eqref{Eq:NotO1} it follows that $\mathbb{P}(F^{H})\rightarrow 1$ and again the result of the Lemma holds.
\end{proof}

\makeatletter%
\if@twocolumn%
\begin{table*}
\caption{Two-step community detection algorithm}\label{Alg.1}
\centering
\begin{tabular}{|p{6in}|}
\hline
1: Start with graph $G$ and side information $\boldsymbol{y}$\\

2: Generate an Erd\"{o}s-Renyi graph $H_{1}$ with edge probability $\frac{D}{\log(n)}$. Use it to partition $G$ into
$G_{1} = G \cap H_{1}$ and $G_{2} = G\cap H_1^c$.\\

3: Apply weak recovery algorithm~\cite{partial_alg} on $G_{1}$, calling the resulting communities $A'/B'$.\\

4:  Initialize $\tilde{A}\leftarrow A'$ and $\tilde{B}\leftarrow B'$. \\

5: For every node $i$ modify $\tilde{A}$ and $\tilde{B}$ as follows:\\

 \quad Flip membership if $i \in \tilde{A}$ and  $E_{G_2}[i,\tilde{B}] \geq E_{G_2}[i,\tilde{A}] + \frac{c}{T}y_{i}$\\
 \quad Flip membership if $i \in \tilde{B}$ and $E_{G_2}[i,\tilde{A}] \geq E_{G_2}[i,\tilde{B}] - \frac{c}{T}y_{i}$\\

6: Check size of communities. If $|A'|\neq|\tilde{A}|$ or equivalently $|B'|\neq|\tilde{B}|$, discard changes via $\tilde{A}\leftarrow A'$ and $\tilde{B}\leftarrow B'$.\\
\hline
\end{tabular}
\end{table*}
\fi%
\makeatother%

The following lemma completes the proof of Theorem~\ref{Th.1}.
\begin{lemma}\label{Le.7}
For sufficiently large $n$, $\mathbb{P}(F_{i}^{H}) > \frac{\log^{3}(n)}{n} \log(\frac{1}{\delta})$ for $\delta \in(0,1)$, if one of the following is satisfied:
$$
\begin{cases}
(\sqrt{a} - \sqrt{b})^{2} < 2 \text{ when } c = o(\log(n)) \nonumber    \\ 
\eta(a,b,\beta)  < 2 \text{ when } c= (\beta+o(1))\log(n),  0<\beta <\frac{T(a-b)}{2} \nonumber \\ 
\beta <1 \qquad\quad \text{ when } c= (\beta+o(1))\log(n), \beta > \frac{T(a-b)}{2}  \nonumber
\end{cases}
$$
\end{lemma}

\begin{proof}
See Appendix~\ref{sec:app1-1}.
\end{proof}
Combining Lemmas~\ref{Le.4},~\ref{Le.5},~\ref{Le.6},~\ref{Le.7} concludes the proof of the theorem.
\end{proof}

\subsection{Sufficient Conditions}
\label{Eff.1}

Sufficient conditions are derived via a two-step algorithm whose first step uses a component from~\cite{partial_alg}, a method based on spectral properties of the graph that achieves weak recovery. 
%

We start with an independently generated random graph $H_1$ built on the same $n$ nodes where each candidate edge has probability $\frac{D}{\log(n)}$. The complement of $H_1$ is denoted $H_2$. Then $G$ is partitioned as follows: $G_{1} = G \cap H_{1}$ and $G_{2} = G \cap H_{2}$. $G_1$ will be used for the weak recovery step, $G_2$ for local modification. The partitioning of $G$ allows the two steps to remain independent.

We perform a weak recovery algorithm~\cite{partial_alg} on $G_{1}$. Since $G_{1}$ is a graph  with connectivity parameters $(\frac{Da}{n},\frac{Db}{n})$, the weak recovery algorithm is guaranteed to return two communities  $A^{'}$, $B^{'}$ that agree with the true communities $A$, $B$ on at least $(1-\delta(D))n$ nodes so that $\lim_{D\rightarrow\infty}\delta(D)= 0$ (i.e., weak recovery). A sufficient condition for that to happen~\cite{partial_alg}, e.g., is $D=O(\log\log n)$. 

The community assignments are locally modified as follows: for a node $i \in A^{'}$, flip its membership if the number of $G_2$ edges between $i$ and $B^{'}$ is greater than or equal the number of $G_2$ edges between $i$ and $A^{'}$ plus $\frac{c}{T}y_{i}$. For node $j \in B^{'}$, flip its membership if the number of $G_2$ edges between $j$ and $A^{'}$ is greater than or equal the number of $G_2$ edges between $j$ and $B^{'}$ minus $\frac{c}{T}y_{j}$. If the number of flips in the two clusters are not the same, keep the clusters unchanged. The detailed algorithm is shown in Table~\ref{Alg.1}.

\makeatletter%
\if@twocolumn%
\else%
\begin{table*}
\caption{Algorithm for exact recovery.}\label{Alg.1}
\centering
\begin{tabular}{|p{6in}|}
\hline
Algorithm 1\\
\hline
1: Start with graph $G$ and side information $\boldmath y$\\

2: Generate an Erd\"{o}s-Renyi graph $H_{1}$ with edge probability $\frac{D}{\log(n)}$. Use it to partition $G$ into
$G_{1} = G \cap H_{1}$ and $G_{2} = G\cap H_1^c$.\\

3: Apply weak recovery algorithm~\cite{partial_alg} on $G_{1}$, calling the resulting communities $A'/B'$.\\

4:  Initialize $\tilde{A}\leftarrow A'$ and $\tilde{B}\leftarrow B'$. \\

5: For each node $i$ modify $\tilde{A}$ and $\tilde{B}$ as follows:\\

 \quad Flip membership if $i \in \tilde{A}$ and  $E_{G_2}[i,\tilde{B}] \geq E_{G_2}[i,\tilde{A}] + \frac{c}{T}y_{i}$\\
 \quad Flip membership if $i \in \tilde{B}$ and $E_{G_2}[i,\tilde{A}] \geq E_{G_2}[i,\tilde{B}] - \frac{c}{T}y_{i}$\\

6: Check size of communities. If $|A'|\neq|\tilde{A}|$ or equivalently $|B'|\neq|\tilde{B}|$, discard changes via $\tilde{A}\leftarrow A'$ and $\tilde{B}\leftarrow B'$.\\
\hline
\end{tabular}
\end{table*}
\fi%
\makeatother%

\begin{theorem}\label{Le.Eff.1}
With probability approaching one as $n$ grows, the algorithm above successfully recovers the communities if:
\[
  \begin{cases}
(\sqrt{a} - \sqrt{b})^{2}  > 2, \text{ when } c = o(\log(n)) \\
\eta(a,b,\beta) > 2 \text{ when }  c = (\beta+o(1))\log(n), 0<\beta <\frac{T(a-b)}{2} \\
   \beta>1 \text{ when }  c = (\beta+o(1))\log(n), \beta >\frac{T(a-b)}{2} 
    \end{cases}
  \]
\end{theorem}

\begin{proof}
We first upper bound the misclassification probability of a node assuming $H_{2}$ is a complete graph, then adjust the bound to account for the departure of $H_2$ from a complete graph.

\begin{figure}
\centering
\includegraphics[width=2.0in]{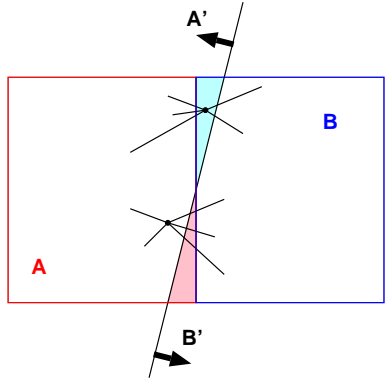}
\caption{Two types of error events for the two-stage algorithm. The node in the top half of the figure is misclassified in weak recovery, and remains uncorrected via local modification. The node at the bottom half is correctly classified in weak recovery, but is mistakenly flipped by local modification.}
\label{fig:mis-class}
\end{figure}

Fig.~\ref{fig:mis-class} shows the mis-classification conditions: an error happens either when the weak recovery was correct and is overturned by the local modification, or when the weak recovery is incorrect and is {\em not} corrected by local modification. Let $W \sim Bern(p)$ and $Z \sim Bern(q)$ represent edges inside a community and across communities, respectively. Let $y_{i} \in \{1,-1\}$ with probabilities $(1-\alpha),\alpha$, respectively. For simplicity, we will write $\delta$ instead of $\delta(D)$. Then, the mis-classification probability is:
\begin{align}
P_{e} & = \mathbb{P}\big(\text{node $i$ is mislabeled}\big) \nonumber\\ 
& = \mathbb{P}\bigg(\sum_{k=1}^{(1-\delta)\frac{n}{2}} Z_{k} + \sum_{k=1}^{\delta\frac{n}{2}} W_{k} \geq \sum_{j=1}^{(1-\delta)\frac{n}{2}} W_{j} + \sum_{j=1}^{\delta\frac{n}{2}} Z_{j} + \frac{c}{T}y_{i}\bigg)
\label{eq.Eff.1}
\end{align}
To adjust for the fact that $H_2$ is not complete, the following Lemma is used, noting that $H_2=H_1^c$.
\begin{lemma}\label{Le.7.1}
With high probability, the degree of any node in $H_{1}$ is at most $\frac{2Dn}{\log(n)}$.
\end{lemma}

\begin{proof}
Let $\{Y_{i}\}_{i=1,\cdots,n}$ be a sequence of i.i.d. Bernoulli random variables with parameter $\frac{D}{\log(n)}$. Define $Y = \sum_{i=1}^{n-1} Y_{i}$. Then, $\mathbb{E}[Y] = \frac{Dn}{\log(n)}$ and hence, by Chernoff bound:
\begin{align}
\mathbb{P}(Y \geq \frac{2Dn}{\log(n)}) & \leq e^{-\frac{1}{4} \frac{D}{\log(n)} n}
\end{align}
Thus, by using a union bound:
\begin{align}
\mathbb{P}\bigg(\exists  \text{ a node degree} > \frac{2Dn}{\log(n)} \bigg) & \; \leq \; n \mathbb{P}\bigg(Y \geq \frac{2Dn}{\log(n)}\bigg) \nonumber\\
&\leq e^{-\frac{1}{4} \frac{D}{\log(n)} n + \log(n)} \underset{n\rightarrow\infty}{\longrightarrow} 0 \nonumber
\end{align}
\end{proof}

Having bounded from below the degree of $H_2$, the correct error probability (for the incomplete $H_2$) can be arrived at by removing no more than $\frac{2D}{\log(n)}n$ terms from the summations on the right hand side of \eqref{eq.Eff.1}. If we remove exactly $\frac{2D}{\log(n)}n$ terms, the following upper bound on error probability is obtained:
\begin{align}\label{eq.Eff.2}
P_{e} &\leq \mathbb{P}\bigg(\sum_{k=1}^{(1-\delta)\frac{n}{2}} Z_{k} + \sum_{k=1}^{\delta\frac{n}{2}} W_{k} \geq \twocolbreak
\sum_{j=1}^{(1-\delta)\frac{n}{2}-\frac{2D}{\log(n)}n} W_{j} + \sum_{j=1}^{\delta\frac{n}{2}-\frac{2D}{\log(n)}n} Z_{j} + \frac{c}{T}y_{i}\bigg )
\end{align}
The following lemma shows an upper bound on $P_{e}$.

\begin{lemma}\label{Le.13}
\begin{align*}
P_{e} \leq
\begin{cases}  n^{-\frac{1}{2} (\sqrt{a}-\sqrt{b})^{2} + o(1)} + n^{-(1+\Omega(1))}\twocolnewline
\hspace{0.5cm} \text{ when } c= o(\log(n))\\
 n^{-\frac{1}{2} \eta(a,b,\beta) + o(1)} + n^{-(1+\Omega(1))} \twocolnewline
\hspace{0.5cm} \text{ when } c = (\beta+o(1))\log(n) \;,\; 0<\beta < \frac{T(a-b)}{2}\\
 n^{-\frac{1}{2} \eta(a,b,\beta) + o(1)} + n^{-\beta} + n^{-(1+\Omega(1))} \twocolnewline
\hspace{0.5cm}\text{ when } c = (\beta+o(1))\log(n) \;,\; \beta > \frac{T(a-b)}{2}
\end{cases}
\end{align*}
\end{lemma}

\begin{proof}
See Appendix~\ref{app.3}.
\end{proof}

A simple union bound yields:
\begin{align}
{\mathbb P}(\text{failure}) \leq
\begin{cases}  
n^{1-\frac{1}{2} (\sqrt{a}-\sqrt{b})^{2} + o(1)},  \twocolnewlineonly
  \hspace{0.8cm} \text{ when } c= o(\log(n)) \\
 n^{1-\frac{1}{2} \eta(a,b,\beta) + o(1)},  \twocolnewlineonly
  \hspace{0.8cm} \text{ when } c = \beta\log(n) \;,\; 0<\beta < \frac{T(a-b)}{2}\\
  n^{1-\beta + o(1)}, \hspace{0.14cm} \text{ when } c = \beta\log(n) \;,\; \beta > \frac{T(a-b)}{2} 
\end{cases}
\label{eq.232}
\end{align}
For the last case, $\beta>1$ remains sufficient because of the following lemma.
\begin{lemma}\label{Le.11}
$\beta > 1 \Rightarrow \eta > 2$.
\end{lemma}

\begin{proof}
Let $a + b - \beta - 2\frac{\gamma}{T} + \frac{\beta}{T}\log(\frac{\gamma+\beta}{\gamma-\beta}) = \psi(a,b,\beta)$. Then, from the definition of $\eta$:
\begin{equation}\label{lem.ap3}
\eta(a,b,\beta) - 2\beta = \psi(a,b,\beta)
\end{equation}
Since $\psi(a,b,\beta)$ is convex in $\beta$, it can be shown that at the optimal $\beta^{*}$, $\log(\frac{\gamma^{*}+\beta^{*}}{\gamma^{*}-\beta^{*}}) = T$. Using this fact and substituting in~\eqref{lem.ap3}:
\begin{equation}\label{lem.ap4}
\eta(a,b,\beta) - 2\beta \geq a + b - 2\frac{\gamma^{*}}{T}
\end{equation}
By the definition of $\gamma$: $\gamma + \beta = \frac{abT^{2}}{\gamma - \beta}$. Using the fact that $\frac{\gamma^{*} + \beta^{*}}{\gamma^{*} - \beta^{*}} = \frac{a}{b}$ leads to $\frac{a}{b} = \frac{abT^{2}}{(\gamma^{*} - \beta^{*})^{2}}$, which implies that $\gamma^{*} = bT + \beta^{*}$. Hence, by substituting in~\eqref{lem.ap4}:
\begin{equation}\label{eq.550}
\eta(a,b,\beta) - 2\beta \geq a - b - 2\frac{\beta^{*}}{T} 
\end{equation}
Also, it can be shown that at $\beta^{*}$, $\gamma^{*}= \beta^{*}(\frac{a+b}{a-b})$. This implies that $\beta^{*} = \frac{T(a-b)}{2}$. Substituting in~\eqref{eq.550} leads to: $\eta(a,b,\beta) - 2\beta \geq 0$, which implies that $\eta > 2$ when $\beta > 1$.
\end{proof}

Combining the last lemma with~\eqref{eq.232} concludes the proof.
\end{proof}


\section{Partially Revealed Labels}\label{sec2}
In this section, we consider side information consisting of partially revealed labels, where $\epsilon \in (0,1)$ is the proportion of labels that remains unknown despite the side information. Tight necessary and sufficient conditions are presented for exact recovery under this type of side information. 
Similar to the noisy label side information, we begin by expressing the log-likelihood function. For a given side information vector $\boldsymbol{y}$, $\mathbb{P}(\boldsymbol{y}|\boldsymbol{x})=0$ if a label contradicts the side information.\footnote{We say a label contradicts the side information if the side information is not an erasure and it disagrees with the label.} All label vectors $\boldsymbol{x}$ that do not contradict side information and satisfy the balanced prior, have the same conditional probability. Thus, for all $\boldsymbol{x}$ that have non-zero conditional probability, the log-likelihood function can be written as:
\begin{align}\label{main.rule.partial.noisy}
\nonumber \log\big(\mathbb{P}(G,\boldsymbol{y}|\boldsymbol{x})\big)  \overset{(a)}{=} & \log\big(\mathbb{P}(G|\boldsymbol{x})\big) + \log\big(\mathbb{P}(\boldsymbol{y}|\boldsymbol{x})\big) \\ 
\overset{(b)}{=} & R + T\big(E(A)+E(B)\big)(1+o(1))
\end{align}
where $(a)$ holds because $G,\boldsymbol{y}$ are independent given $\boldsymbol{x}$. In $(b)$, all terms that are independent of $\boldsymbol{x}$ have been collected into a constant $R$, and  $\log(\frac{p(1-q)}{q(1-p)})$ has been approximated by $(1+o(1))T$, which is made possible because $(1-p), (1-q)$ both approach $1$ as $n\to\infty$. 

The following lemma shows that if the graph includes at least one pair of nodes that have more connections to the opposite-labels than similar-labels {\em and} if their side information is an erasure, the maximum likelihood detector will fail.

\begin{lemma}\label{Le.s2.1}
Define the following events:
\begin{align}\nonumber
F_{A} &= \{ \exists i \in A: (E[i,B] - E[i,A]) \geq 1 \text{ and } y_{i} = 0\} \\ \nonumber
F_{B} &= \{ \exists j \in B: (E[j,A] - E[j,B]) \geq 1 \text{ and } y_{j} = 0\} 
\end{align}
 Then, $F_{A}\cap F_{B} \Rightarrow F$.
\end{lemma}

\begin{proof}

From the sets $A,B$, we swap the nodes $i,j$, producing $\hat{A} = A\backslash \{i\}\cup\{j\}$ and $\hat{B} = B\backslash \{j\}\cup\{i\}$. We intend to show that subject to observing the graph $G$ and the side information $\boldsymbol y$, the likelihood of $\hat{A}, \hat{B}$ is larger than the likelihood of $A, B$, therefore under the condition $F_A\cap F_B$, maximum likelihood will fail.

Let $A_{ij}\sim Bern(q)$ be a random variable representing the existence of the edge between nodes $i$ and $j$. Then, from~\eqref{main.rule.partial.noisy}:
\begin{align} 
\twocolAlignMarker \log\big(\mathbb{P}(G,\boldsymbol{y}|\hat{A},\hat{B})\big) \twocolnewline
&=  R + T(1+o(1))\big(E(\hat{A})+E(\hat{B})\big) \nonumber \\ 
&=  R + T(1+o(1))\big(E(A)+E(B)\big) + T(1+o(1)) \nonumber \\ 
& \times \big(E[j,A] - E[i,A] -  E[j,B] + E[i,B] -2A_{ij}\big) \nonumber\\
\overset{(a)}{\geq} &  \log\big(\mathbb{P}(G,\boldsymbol{y}|A,B)\big) + 2T(1+o(1))(1-A_{ij}) \nonumber \\
\overset{(b)}{\geq} & \log\big(\mathbb{P}(G,\boldsymbol{y}|A,B)\big)
\end{align}
where $(a)$ holds by the assumption that $F_{A} \cap F_{B}$ happened and $(b)$ holds because $(1-A_{ij}) \geq 0$ and $T \geq 0$. The inequality $(b)$ implies the failure of maximum likelihood.
\end{proof}


\subsection{Necessary Conditions}\label{LB.2}

\begin{theorem}\label{Th.4}
The maximum likelihood failure probability is bounded away from zero if:
\begin{itemize}
\item
$\log(\epsilon) = o(\log(n))$ and $(\sqrt{a} - \sqrt{b})^{2}  < 2$
\item
$\log(\epsilon) = -(\beta+o(1))\log(n)$, $\beta>0$, and $\frac{1}{2}(\sqrt{a} - \sqrt{b})^{2} + \beta < 1$
\end{itemize}
\end{theorem}

\begin{proof}
Let $H$ be a subset of $A$ with $|H| = \frac{n}{\log^{3}(n)}$. Consider the 
following modification to Definition~\ref{Def.1}:
\[
F_{i}^{H} = \big\{  y_{i} = 0 \big\} \cap \big\{  E[i,A\backslash H] + 1 + \frac{\log(n)}{\log\log(n)} \leq E[i,B] \big\}
\]
It is not difficult to show that Lemmas~\ref{Le.4},~\ref{Le.5},~\ref{Le.6} remain valid under this modification. To complete the proof, it is sufficient to find conditions under which $\mathbb{P}(F_{i}^{H}) > \frac{\log^{3}(n)}{n} \log(\frac{1}{\delta})$ asymptotically (in $n$) for all $\delta \in (0,1)$.
\begin{lemma}\label{Le.8}
For sufficiently large $n$, $\mathbb{P}(F_{i}^{H}) > \frac{\log^{3}(n)}{n} \log(\frac{1}{\delta})$ for $\delta \in(0,1)$, if one of the following is satisfied:
\begin{equation*}
\begin{cases}
(\sqrt{a} - \sqrt{b})^{2}  < 2, \hspace{0.1cm} \text{ when } \log(\epsilon) = o(\log(n)) \nonumber  \\
(\sqrt{a} - \sqrt{b})^{2} + 2\beta < 2, \twocolnewline
\hspace{1cm} \text{ when } \log(\epsilon)= -(\beta+o(1))\log(n), \beta>0  
\end{cases}
\end{equation*}
\end{lemma}

\begin{proof}
See Appendix~\ref{sec:app1-2}.
\end{proof}
Combining Lemma~\ref{Le.8} with the modified form of Lemmas~\ref{Le.4},~\ref{Le.5}, and~\ref{Le.6}, concludes the proof of the theorem.
\end{proof}

\subsection{Sufficient Conditions}\label{UB.2}


This section shows sufficient conditions for exact recovery by introducing an algorithm whose exact recovery conditions are identical to Section~\ref{LB.2}. The first stage of the algorithm is the same as Section~\ref{Eff.1}. The second stage involving local modification is new and is described below.

The community assignments are locally modified for each node $i$ as follows: (a) if $A'/B'$ membership contradicts side information $y_i$, flip node membership or (b) if $y_{i}=0$, re-assign membership of $i$ to the community $A'/B'$ to which it is connected with more edges. After going through all nodes, if the the number of flips in two communities $A', B'$ are not the same, void all local modifications.

\begin{theorem}\label{Le.Eff.2}
The algorithm described above successfully recovers the communities with high probability if:
$$
\begin{cases}
(\sqrt{a} - \sqrt{b})^{2}  > 2, \hspace{0.1cm} \text{when } \log(\epsilon) = o(\log(n)) \nonumber  \\ 
(\sqrt{a} -\sqrt{b})^{2} + 2\beta > 2, \twocolnewline
\hspace{1cm} \text{when } \log(\epsilon) = -(\beta+o(1))\log(n), \beta>0  
\end{cases}
$$ 
\end{theorem}

\begin{proof}
Let $P_{e} = \mathbb{P}(\text{node $i$ to be misclassified})$. Following the same analysis as in the proof of Theorem~\ref{Le.Eff.1}:
\begin{equation}\label{eq.Effs2.2}
P_{e} \leq \epsilon \mathbb{P}\bigg(\sum_{k=1}^{(1-\delta)\frac{n}{2}} Z_{k} + \sum_{k=1}^{\delta\frac{n}{2}} W_{k} \geq \!\!\!\!\!\! \sum_{j=1}^{(1-\delta)\frac{n}{2}-\frac{2D}{\log(n)}n} \!\!\!\!\!\! W_{j} + \!\!\!\!\!\! \sum_{j=1}^{\delta\frac{n}{2}-\frac{2D}{\log(n)}n} \!\!\!\!\!\! Z_{j}  \bigg)
\end{equation}

Using Lemma~\ref{Le.13} and strengthening $c= o(\log(n))$ to $c=0$, equation~\eqref{eq.Effs2.2} can be upper bounded as follows:
\begin{align}\label{eq.Effs2.3}
P_{e} & \leq \epsilon n^{-\frac{1}{2} (\sqrt{a}-\sqrt{b})^{2}} + n^{-(1+\Omega(1))}
\end{align}
Thus, according to asymptotic behavior of $\epsilon$:
\begin{align*}
P_{e} \leq
\begin{cases}  n^{-\frac{1}{2} (\sqrt{a}-\sqrt{b})^{2} + o(1)} + n^{-(1+\Omega(1))}, \nonumber \twocolnewline
\qquad\qquad \text{ when } \log(\epsilon)= o(\log(n)) \nonumber \\
 n^{-\frac{1}{2} (\sqrt{a}-\sqrt{b})^{2} - \beta} + n^{-(1+\Omega(1))}, \nonumber \twocolnewline
\qquad\qquad \text{ when } \log(\epsilon) = -(\beta+o(1))\log(n) \;,\; \beta > 0
\end{cases}
\end{align*}
A simple union bound yields:
\begin{align*}
{\mathbb P}(\text{failure}) \leq
\begin{cases}  n^{1-\frac{1}{2} (\sqrt{a}-\sqrt{b})^{2} + o(1)}, \hspace{0.06cm} \text{ when } \log(\epsilon)= o(\log(n)) \nonumber \\
 n^{1-\frac{1}{2} (\sqrt{a}-\sqrt{b})^{2} - \beta + o(1)} \twocolnewline
\hspace{0.3cm} \text{ when } \log(\epsilon) = -(\beta+o(1))\log(n) \;,\; \beta > 0
\end{cases}
\end{align*}
\end{proof}

\section{More General Side Information}
\label{sec:general}

We now generalize the side information random variable such that each node observes $K$ features (side information) each has arbitrary fixed and finite cardinality $M_{k}, k \in \{1,\cdots, K\}$. The alphabet for each feature $k$ is denoted with $\{u_{1}^{k}, u_{2}^{k}, \cdots, u_{M_{k}}^{k}\}$. Denote, for each node $i$ and feature $k$, $\mathbb{P}(y_{i,k} = u_{m_{k}}^{k} | x_{i} = 1) = \alpha_{+,m_{k}}^{k}$ and $\mathbb{P}(y_{i,k} = u_{m_{k}}^{k} | x_{i} = -1) = \alpha_{-,m_{k}}^{k}, m_{k}\in\{1,\cdots,M_{k}\}$, where $\alpha_{+,m_{k}}^{k}\geq 0$, $\alpha_{-,m_{k}}^{k}\geq 0$ and $\sum_{m_{k}=1}^{M_{k}} \alpha_{+,m_{k}}^{k} = \sum_{m_{k}=1}^{M_{k}} \alpha_{-,m_{k}}^{k} = 1$ for all $k \in\{1,\cdots,K\}$. All features are assumed to be independent conditioned on the labels.
We first consider the case where $K$ is fixed while $\alpha_{+,m_{k}}^{k}$ and $\alpha_{-,m_{k}}^{k}$ are varying with $n$ for $m_{k} \in \{1,\cdots, M_{k}\}$ and $k \in \{1,\cdots, K\}$. To ensure that the quality of the side information is increasing with $n$, assume that $\alpha_{+,m_{k}}^{k}$ and $\alpha_{-,m_{k}}^{k}$ for $m_{k} \in \{1,\cdots, M_{k}\}$ and $k \in \{1,\cdots, K\}$ are constant or monotonic in $n$. Second, we consider the case where $K$ is varying with $n$ while $\alpha_{+,m_{k}}^{k}$ and $\alpha_{-,m_{k}}^{k}$ are fixed for $m_{k} \in \{1,\cdots, M_{k}\}$ and $k \in \{1,\cdots, K\}$. To ensure that the quality of the side information is increasing with $n$, assume that $K$ is non-decreasing with $n$. Necessary and sufficient conditions for exact recovery that are tight except for one special case are provided. 

First the log-likelihood function is presented. For feature $k$, let the number of $\{i\in A: y_{i,k}= u_{m_{k}}^{k}\}$ and $\{i\in B: y_{i,k}= u_{m_{k}}^{k}\}$ be $J_{u_{m_{k}}^{k}}(A)$ and $J_{u_{m_{k}}^{k}}(B)$, respectively. Then, by using similar ideas as in~\eqref{main.rule.Full.noisy}:
\begin{align}
& \log\big(\mathbb{P}(G,\boldsymbol{y_{1}},\boldsymbol{y_{2}},\cdots,\boldsymbol{y_{K}}|\boldsymbol{x})\big)= \twocolbreak
 R + T\big(E(A)+ E(B)\big)(1+ o(1)) + \nonumber \\
& \sum_{k=1}^{K} \sum_{m_{k}=1}^{M_{k}} J_{u_{m_{k}}^{k}}(A) \log\big( \alpha_{+,m_{k}}^{k}\big) + J_{u_{m_{k}}^{k}}(B) \log\big(\alpha_{-,m_{k}}^{k}\big)
\end{align}
\begin{definition}
The side information LLR for outcome $m_k$ of feature $k$ is denoted:
\[
\hfeat \triangleq \log(\frac{\alpha_{+,m_{k}}^{k}}{\alpha_{-,m_{k}}^{k}})
\]
The LLR produced by the side information for each node $i$ is a random variable which we denote with $\hbar_i$ where
$\hbar_i =\sum_k \hbar_{ik}$,
and $\hbar_{ik}$ is the LLR of feature $k$ for node $i$. 
\end{definition}

\begin{lemma}\label{Le.s3.1}
Define the following events:
\begin{align}\nonumber
F_{A} &= \{ \exists i \in A: T(E[i,B] - E[i,A]) - \hbar_i \geq T\} \\ \nonumber
F_{B} &= \{ \exists j \in B: T(E[j,A] - E[j,B]) + \hbar_j \geq T\}  \nonumber
\end{align}
 Then, $F_{A} \cap F_{B} \Rightarrow F$.
\end{lemma}
\begin{proof}
The proof is similar to Lemmas~\ref{Le.1} and ~\ref{Le.s2.1}.
\end{proof}
\subsection{Fixed Number of Features, Variable Quality}
In this section, the number of features $K$ is assumed to be fixed and we show how noisy the outcomes of the features should be so that side information changes the phase transition threshold of exact recovery. 
We begin with $K=1$, i.e. one feature with $M$ outcomes. For each side information outcome $m \in \{1,\cdots,M\}$, two quantities affect the phase transition: the log-likelihood ratio $h_{m} = \log(\frac{\alpha_{+,m}}{\alpha_{-,m}})$ and the conditional probability $\alpha_{\pm,m}$. An outcome is called \textit{informative} if $h_{m} = O(\log(n))$\footnote{We say $h_m = O(\log n)$ when there exists a strictly positive constant $C$ such that $h_m < C \log(n)$ for all sufficiently large $n$.} and \textit{non-informative} if $h_{m} = o(\log(n))$. Also, an outcome is called \textit{rare} if $\log(\alpha_{\pm,m}) = O(\log(n))$ and \textit{not rare} if $\log(\alpha_{\pm,m}) = o(\log(n))$. Hence, four different combinations are possible. The \textit{worst} case is when the outcome is both \textit{non-informative} and \textit{not rare} for both communities, e.g. noisy labels with $\alpha = \frac{1}{\log(n)}$. We will show that if such an outcome exists, then side information will not improve the phase transition threshold. The \textit{best} case is when the outcome is \textit{informative}, and \textit{rare} for one community but \textit{not rare} for the other. This happens, e.g., under noisy label side information with $\alpha = n^{-\beta + o(1)}$.  We have two cases in between: (1) an outcome that is \textit{non-informative} and \textit{rare} for both communities, e.g. partial label reveal side information with $\epsilon = n^{-\beta+o(1)}$ and (2) an outcome that is \textit{informative} and \textit{not rare} for both communities. The last three cases can affect the phase transition threshold under certain conditions. As shown by Theorem~\ref{Th.3}, phase transition is characterized by (the evolution of) the following functions of the statistics of side information.
\begin{align}
f_1(n) &\triangleq \sum_{k=1}^K \hfeat, \\
f_2(n) &\triangleq\sum_{k=1}^K \log(\alpha_{+,m_{k}}^{k}), \\
f_3(n) &\triangleq\sum_{k=1}^K \log(\alpha_{-,m_{k}}^{k})
\end{align}
In the following, the side information outcomes $[u_{m_1}^1, \ldots,u_{m_K}^K]$ are represented by their index $[m_1,\ldots,m_K]$ without loss of generality. Throughout, dependence on $n$ of outcomes and their likelihood is implicit.
\begin{theorem}\label{Th.3}
Assume $\alpha_{+,m_{k}}^{k}$ and $\alpha_{-,m_{k}}^{k}$  are either constant or monotonically increasing or decreasing in $n$. Then, necessary and sufficient conditions for exact recovery depend on side information statistics 
in the following manner:
\begin{enumerate}
\item If there exists any sequence (over $n$) of  side information outcomes $[m_1,\ldots,m_K]$ such that $f_1(n)$, $f_2(n)$, $f_3(n)$ are all $o(\log(n))$,
then $(\sqrt{a}-\sqrt{b})^{2} > 2$ must hold.

\item If there exists any sequence (over $n$) of side information outcomes $[m_1,\ldots,m_K]$ such that $f_1(n)= o(\log(n))$ and $f_2(n), f_3(n)$ evolve according to $-\beta\log(n)+ o(\log(n))$ with $\beta > 0$, then $(\sqrt{a}-\sqrt{b})^{2} + 2\beta > 2$  must hold.

\item
\label{Th3.item3}  If there exists any sequence (over $n$) of side information outcomes $[m_1,\ldots,m_K]$ such that $f_1(n) = \beta_{1}\log(n)+ o(\log(n))$ with $|\beta_{1}| < T\frac{(a-b)}{2}$ and furthermore $f_2(n)=o(\log(n))$ if $\beta_1>0$ and  $f_3(n)=o(\log(n))$ if $\beta_1<0$, then $\eta(a,b,|\beta_{1}|) > 2$ must hold.

\item
\label{Th3.item4} If  there exists any sequence (over $n$) of side information outcomes $[m_1,\ldots,m_K]$ such that $f_1(n) = \beta_{2}\log(n)+ o(\log(n)),  |\beta_{2}| < T\frac{(a-b)}{2}$ and furthermore $f_2(n)= -\beta_{2}^{'}\log(n)+ o(\log(n))$ if $\beta_2>0$ and $f_3(n)= -\beta_{2}^{'}\log(n)+ o(\log(n))$ if $\beta_2<0$, then $\eta(a,b,|\beta_{2}|) + 2\beta_{2}^{'}>2$ must hold.
\end{enumerate}
\end{theorem}

\begin{remark}
The four parts in Theorem~\ref{Th.3} are concurrent. For example, if some side information outcome sequences fall under Item~\ref{Th3.item3} and some fall under Item~\ref{Th3.item4}, then the necessary and sufficient condition for exact recovery is $\min (\eta(a,b,|\beta_{1}|),\eta(a,b,|\beta_{2}|) + 2\beta_{2}^{'}) >2$.
\end{remark}

\begin{remark}\label{Remark10}
When there is any sequence of side information outcomes that satisfies $f_1(n) = \beta\log(n) +o(\log(n))$ with $T\frac{(a-b)}{2} < |\beta|$, a sufficient condition easily follows other achievability proofs for Theorem~\ref{Th.3}, but a matching converse for this case remains unavailable. 
\end{remark}

\begin{proof}
{\bf Converse:} Unlike previous sections, the side information might not be symmetric. Hence, we need to define the events of Section~\ref{LB.1} for both communities $A$ and $B$. Let $H_{1}$ and $H_{2}$ be subsets of the true communities $A$ and $B$, respectively, with $|H_{1}| = |H_{2}| = \frac{n}{\log^{3}(n)}$. 

\begin{definition}
\label{Def.3}
Define the following events for nodes $i\in H_1$:
\begin{align*}
\Delta_{i}^1 &= \big\{ E[i,H_1] \leq \frac{\log(n)}{\log\log(n)}\big\} \nonumber \\
F_{i}^{H_{1}} &= \bigg\{ T E[i,A\backslash H_{1}] + \hbar_i + T + T \frac{\log(n)}{\log\log(n)} \leq T E[i,B] \bigg\} 
\end{align*}
and the following events for nodes $j\in H_2$:
\begin{align*}
\Delta_j^2 &= \big\{ E[j,H_2] \leq \frac{\log(n)}{\log\log(n)}\big\} \nonumber \\
F_{j}^{H_2} &= \bigg\{ T E[j,B\backslash H_2] - \hbar_j + T + T \frac{\log(n)}{\log\log(n)} \leq T E[j,A] \bigg\} 
\end{align*}
and the following overall events:
\begin{align*}
\Delta^1 &= \cap_{i\in H_1} \Delta_{i}^1  \quad&\quad \Delta^2 &= \cap_{j\in H_2} \Delta_{j}^2 \nonumber \\ 
F^{H_1} &= \cup_{i \in H_1} F_i^{H_1} \quad&\quad F^{H_2} &= \cup_{j \in H_2} F_j^{H_2} 
\end{align*}
\end{definition}

Lemmas~\ref{Le.4},~\ref{Le.5},~\ref{Le.6} remain valid according to Definition~\ref{Def.3}. It remains to show under which conditions $\mathbb{P}(F_{i}^{H_{1}}) > \frac{\log^{3}(n)}{n} \log(\frac{1}{\delta})$ and $\mathbb{P}(F_{j}^{H_{2}}) > \frac{\log^{3}(n)}{n} \log(\frac{1}{\delta})$, asymptotically for all $\delta \in(0,1)$.

\begin{lemma}\label{Le.9}
Both $\mathbb{P}(F_{i}^{H_{1}})$ and $\mathbb{P}(F_{j}^{H_{2}})$ are greater than $\frac{\log^{3}(n)}{n} \log(\frac{1}{\delta})$, $\delta \in(0,1)$ for sufficiently large $n$ if at least one of the following conditions holds:
\begin{itemize}
\item If there exists a sequence (over $n$) of  side information outcomes $[m_1,\ldots,m_K]$ such that $f_1(n)$, $f_2(n)$, $f_3(n)$  are all $o(\log(n))$ and concurrently $(\sqrt{a}-\sqrt{b})^{2}  < 2$.
\item If there exists a sequence (over $n$) of  side information outcomes $[m_1,\ldots,m_K]$ such that $f_1(n) =  o(\log(n))$, and $f_2(n),f_3(n)$ evolve according to $-\beta\log(n)+ o(\log(n)), \beta > 0$, and concurrently $(\sqrt{a}-\sqrt{b})^{2} + 2\beta  < 2$.
\item If there exists a sequence (over $n$) of  side information outcomes $[m_1,\ldots,m_K]$ such that $f_1(n) = \beta_1\log(n)+ o(\log(n)),  \quad |\beta_1| < T\frac{(a-b)}{2}$, and furthermore $f_2(n)=o(\log(n))$ if $\beta_1>0$ and  $f_3(n)=o(\log(n))$ if $\beta_1<0$, and concurrently  $\eta(a,b,|\beta_1|)  < 2$ .
\item If  there exists a sequence (over $n$) of  side information outcomes $[m_1,\ldots,m_K]$ such that $f_1(n) = \beta_2\log(n)+ o(\log(n)),  \quad |\beta_{2}| < T\frac{(a-b)}{2}$ and furthermore $f_2(n)= -\beta_{2}'\log(n)+ o(\log(n))$ if $\beta_2>0$ and $f_3(n)= -\beta_{2}'\log(n)+ o(\log(n))$ if $\beta_2<0$, and concurrently $\eta(a,b,|\beta_2|) + \beta_2'  < 2$.
\end{itemize}
\end{lemma}

\begin{proof}
Please see Appendix~\ref{sec:app1-5}
\end{proof}
Combining Lemma~\ref{Le.9} with Lemmas~\ref{Le.4},~\ref{Le.5}, and~\ref{Le.6} modified according to Definition~\ref{Def.3}, concludes the proof of converse.

%

{\bf Achievability:} Achievability of Theorem~\ref{Th.3} is proven via an algorithm whose exact recovery conditions are identical to the necessary conditions provided in Lemma~\ref{Le.9}. The first stage of the algorithm is the same as Section~\ref{Eff.1}. After the first stage, we have $G_{2}$, the side information $\boldsymbol{y}_{1},\cdots,\boldsymbol{y}_{K}$, $A^{'}$ and $B^{'}$. Locally modify the community assignment as follows: for a node $i \in A^{'}$, flip its membership if $E[i,B'] \ge E[i,A']+\frac{\hbar_i}{T}$ and for node $j \in B^{'}$, flip its membership if  $E[j,A'] \ge E[j,B']-\frac{\hbar_j}{T}$. If the the number of flips in each cluster is not the same, keep the clusters unchanged.

\begin{lemma}\label{Le.Eff.3}
The algorithm described above successfully recovers the communities with high probability if the following are satisfied simultaneously:
\begin{itemize}
\item If there exists a sequence (over $n$) of  side information outcomes $[m_1,\ldots,m_K]$ such that $f_1(n), f_2(n), f_3(n)$ are all $o(\log(n))$ and concurrently $(\sqrt{a}-\sqrt{b})^{2}  > 2$.
\item If there exists a sequence (over $n$) of  side information outcomes $[m_1,\ldots,m_K]$ such that $f_1(n) =  o(\log(n))$,  and
$f_2(n), f_3(n)$ evolve according to $-\beta\log(n)+ o(\log(n)), \beta > 0$ and concurrently $(\sqrt{a}-\sqrt{b})^{2} + 2\beta  > 2$.
\item If  there exists a sequence (over $n$) of  side information outcomes $[m_1,\ldots,m_K]$ such that $f_1(n) = \beta_1\log(n)+ o(\log(n)),  \quad |\beta_1| < T\frac{(a-b)}{2}$, and furthermore $f_2(n)=o(\log(n))$ if $\beta_1>0$ and  $f_3(n)=o(\log(n))$ if $\beta_1<0$, and concurrently $\eta(a,b,|\beta_1|)  > 2$.
\item If there exists a sequence (over $n$) of  side information outcomes $[m_1,\ldots,m_K]$ such that $f_1(n) = \beta_2\log(n)+ o(\log(n)),  \quad |\beta_{2}| < T\frac{(a-b)}{2}$, and furthermore $f_2(n)= -\beta_{2}'\log(n)+ o(\log(n))$ if $\beta_2>0$ and $f_3(n)= -\beta_{2}'\log(n)+ o(\log(n))$ if $\beta_2<0$, and concurrently $\eta(a,b,|\beta_2|) + \beta_2'  > 2$.
\end{itemize}
\end{lemma}

\begin{proof}
Define $P_{e} = \mathbb{P}(\text{node $i$ to be misclassified})$. Following similar analysis as in the proof of Lemma~\ref{Le.Eff.1} leads to:
\begin{align}
& P_{e} \leq  \frac{1}{2} \bigg( n^{-1-\Omega(1)} + \sum_{m_{1}=1}^{M_{1}}\sum_{m_{2}=1}^{M_{2}}\cdots \sum_{m_{K}=1}^{M_{K}} \prod_{k=1}^{K} (\alpha_{+,m_{k}}^{k}) \twocolbreak \times \mathbb{P}\bigg(\sum_{l=1}^{\frac{n}{2}} (Z_{l} - W_{l}) \geq \sum_{k=1}^{K}\frac{\hfeat}{T} + \psi_{n}\log(n) \bigg) \bigg)  +   \nonumber\\
& \frac{1}{2} \bigg( n^{-1-\Omega(1)} + \sum_{m_{1}=1}^{M_{1}}\sum_{m_{2}=1}^{M_{2}}\cdots \sum_{m_{K}=1}^{M_{K}} \prod_{k=1}^{K} (\alpha_{-,m_{k}}^{k}) \twocolbreak \times \mathbb{P}\bigg(\sum_{l=1}^{\frac{n}{2}} (Z_{l} - W_{l}) \geq -\sum_{k=1}^{K}\frac{\hfeat}{T} + \psi_{n}\log(n) \bigg) \bigg) \label{eq.Effs3.2}
\end{align}
where $\psi_{n} = o(1)$.

Similar to Lemma~\ref{Le.13}, it can be shown that any term inside the nested sum in~\eqref{eq.Effs3.2} is upper bounded by:
\begin{itemize}
\item $n^{-\frac{1}{2} (\sqrt{a}-\sqrt{b})^{2} + o(1)}$ if there exists a sequence (over $n$) of  side information outcomes $[m_1,\ldots,m_K]$ such that $f_1(n), f_2(n), f_3(n)$ are all $o(\log(n))$.

\item $n^{-\frac{1}{2} (\sqrt{a}-\sqrt{b})^{2} - \beta + o(1)}$ if there exists a sequence (over $n$) of  side information outcomes $[m_1,\ldots,m_K]$ such that $f_1(n) =  o(\log(n))$ and $f_2(n), f_3(n)$ evolve according to $-\beta\log(n)+ o(\log(n)), \quad \beta > 0$

\item $n^{-\frac{1}{2}\eta(a,b,|\beta_1|) + o(1)}$ if there exists a sequence (over $n$) of  side information outcomes $[m_1,\ldots,m_K]$ such that $f_1(n) = \beta_1\log(n)+ o(\log(n)),  \quad |\beta_1| < T\frac{(a-b)}{2}$ and furthermore $f_2(n)=o(\log(n))$ if $\beta_1>0$ and  $f_3(n)=o(\log(n))$ if $\beta_1<0$.

\item $n^{-\frac{1}{2}\eta(a,b,|\beta_2|)-\beta_2' + o(1)}$ if there exists a sequence (over $n$) of  side information outcomes $[m_1,\ldots,m_K]$ such that $f_1(n) = \beta_2\log(n)+ o(\log(n)),  \quad |\beta_2| < T\frac{(a-b)}{2}$ and furthermore $f_2(n)= -\beta_{2}'\log(n)+ o(\log(n))$ if $\beta_2>0$ and $f_3(n)= -\beta_{2}'\log(n)+ o(\log(n))$ if $\beta_2<0$.
\end{itemize}
Since $K$ and $\{M_{k}, k=1,\ldots,\K\}$ are fixed, a union bound over the nodes concludes the proof of Lemma~\ref{Le.Eff.3}. 
\end{proof}
This concludes the proof of achievability.
\end{proof}

We now give an example of side information with $K=1$ and fixed cardinality and analyze the effect of the evolution of the distribution of side information with growing $n$. 

Consider the weakly symmetric side information whose transition probability matrix $\mathbb{P}(y|x)$ is defined as follows: every row of the transition matrix $\mathbb{P}(\cdot|x)$ is a permutation of every other row, and all the column sums $\sum_{x} \mathbb{P}(y|x)$ are equal. Since the labels are either $1$ or $-1$, all the column sums are $\frac{2}{M}$. Without loss of generality, assume the first row $\mathbb{P}(y|x=+1)$ is arranged in descending order, i.e. $\mathbb{P}(y_{l+1}|x=+1)\geq \mathbb{P}(y_{l}|x=+1)$, $1\leq l \leq M-1$. Thus, for even $M$ (odd $M$ follow similarly), by the weakly symmetry property of $\mathbb{P}(y|x)$:  $\alpha_{\pm,l} + \alpha_{\pm,M-l+1} = \frac{2}{M}$ and $h_{l} = -h_{M-l+1}$, $1 \leq l \leq \frac{M}{2}$. Thus, if $h_{\frac{M}{2}} = \beta \log(n) + o(\log(n))$, i.e., $h_{\frac{M}{2}} = O(\log(n))$, this implies that $h_{l} = O(\log(n))$ for all $1 \leq l \leq M$, and hence, this maps to the third case of Theorem~\ref{Th.3}. In other words, $\eta(a,b,|\beta|) > 2$ is necessary and sufficient for exact recovery (assuming $|\beta| < \frac{T(a-b)}{2}$). On the other hand, if $h_{\frac{M}{2}}$ is in the order of $o(\log(n))$, this maps to the first case of Theorem~\ref{Th.3}, and hence, side information does not change the exact recovery phase transition.

\subsection{Varying Number of Fixed-Quality Features}
\label{sec:varying}
In this section, $\alpha_{+,m_{k}}^{k}$ and $\alpha_{-,m_{k}}^{k}$ are independent of $n$. We study how many features $K$ are needed so that side information can improve the phase transition threshold of exact recovery. We show that when $K = o(\log(n))$, side information will not improve the phase transition of exact recovery. A direct extension of our result shows that with $K = O(\log(n))$, side information can improve the phase transition, but this result is omitted here both in the interest of brevity and in part because it can be considered a straight forward extension of~\cite[Theorem 4]{Abbe_1} which showed the result in the special case of $K = \log(n)$. 

\begin{theorem}\label{Th.vec.1}
Assume that $M_{k} = M$  and all features are i.i.d. conditioned on the labels. Let $\alpha_{+,m_{k}}^{k}$ and $\alpha_{-,m_{k}}^{k}$ be non-zero and independent of $n$. Then, if $K = o(\log(n))$,  $(\sqrt{a}-\sqrt{b})^{2} > 2$ is necessary and sufficient for exact recovery.
\end{theorem}

\begin{proof}
{\bf Converse:} Using Definition~\ref{Def.3}, it remains to show under what conditions $\mathbb{P}(F_{i}^{H_{1}}) > \frac{\log^{3}(n)}{n} \log(\frac{1}{\delta})$ and $\mathbb{P}(F_{j}^{H_{2}}) > \frac{\log^{3}(n)}{n} \log(\frac{1}{\delta})$ asymptotically  for all $\delta \in(0,1)$.

\begin{lemma}\label{Le.vec.1}
For $K = o(\log(n))$, both $\mathbb{P}(F_{i}^{H_{1}})$ and $\mathbb{P}(F_{j}^{H_{2}})$ are greater than $\frac{\log^{3}(n)}{n} \log(\frac{1}{\delta})$, $\delta \in(0,1)$ for sufficiently large $n$ if $(\sqrt{a}-\sqrt{b})^{2} < 2$.
\end{lemma}

\begin{proof}
Let $W_{i} \sim Bern(p)$, $Z_{i} \sim Bern(q)$. Then, in a manner similar to Lemmas~\ref{Le.9},~\ref{Le.10}:
\begin{align}
& \mathbb{P}(F_{i}^{H_1}) \nonumber \\ 
\geq &  \mathbb{P}\bigg( \sum_{l=1}^{\frac{n}{2}} [Z_{l} - W_{l}] \geq \frac{\hbar_i}{T} + 1 + \frac{\log(n)}{\log\log(n)}\bigg) \nonumber\\
\geq & e^{-\log(n) (1+o(1)) (\sup_t \frac{a+b}{2} - \frac{b}{2} (\frac{a}{b})^{t} - \frac{a}{2} (\frac{a}{b})^{-t} -\frac{K \log(\mathbb{E}_{+}[e^{-t \hbar_{ik}}]) }{\log(n)} ) }  \label{vec.eq.1}
\end{align}
where $\mathbb{E}_{+}[e^{-t \hbar_{ik}}]$ is the moment generating function of the side information LLR, for feature $k$ of node $i$, conditioned on $x_i=1$. Since $K = o(\log(n))$, substituting in~\eqref{vec.eq.1} leads to:
\begin{align}
\mathbb{P}(F_{i}^{H_{1}}) & \geq e^{-\log(n) (1+o(1)) (\sup_{t \in \mathbb{R}} \frac{a+b}{2} - \frac{b}{2} (\frac{a}{b})^{t} - \frac{a}{2} (\frac{a}{b})^{-t} ) } \nonumber \\
& \geq n^{-\frac{1}{2}(\sqrt{a}-\sqrt{b})^{2} + o(1)}
\end{align}
where the last inequality holds by evaluating the supremum. Thus, if $\frac{1}{2}(\sqrt{a}-\sqrt{b})^{2} < 1$, $\mathbb{P}(F_{i}^{H_{1}}) > \frac{\log^{3}(n)}{n} \log(\frac{1}{\delta})$, $\delta \in(0,1)$ for sufficiently large $n$.

Similarly,
\begin{align}
&\mathbb{P}(F_{j}^{H_{2}}) \nonumber \\
\geq &  e^{-\log(n) (1+o(1)) (\sup_{t \in \mathbb{R}} \frac{a+b}{2} - \frac{b}{2} (\frac{a}{b})^{t} - \frac{a}{2} (\frac{a}{b})^{-t} -\frac{K \log(\mathbb{E}_{-}[e^{t \hbar_{jk}}]) }{\log(n)} ) }   \nonumber \\
\geq & n^{-\frac{1}{2}(\sqrt{a}-\sqrt{b})^{2} + o(1)}
\end{align}
Thus, $\frac{1}{2}(\sqrt{a}-\sqrt{b})^{2} < 1$ implies $\mathbb{P}(F_{j}^{H_{2}}) > \frac{\log^{3}(n)}{n} \log(\frac{1}{\delta})$ for all $\delta \in(0,1)$ for sufficiently large $n$.
\end{proof}
Combining Lemmas~\ref{Le.4},~\ref{Le.5},~\ref{Le.6},~\ref{Le.vec.1} concludes the proof of converse.

{\bf Achievability:} It is known that $\frac{1}{2}(\sqrt{a}-\sqrt{b})^{2} > 1$ is sufficient if the only observation was the graph. Combining this with the converse completes the proof.
\end{proof}


%

\appendices

\section{Proof of Lemma~\ref{Le.7}}
\label{sec:app1-1}

Define $l = \frac{n}{2}$ and $\Gamma(t) \triangleq \log(\mathbb{E}_{X}[e^{tx}])$ for a random variable $X$. Then,
\begin{align}
& \mathbb{P}(F_{i}^{H})  = \mathbb{P}\bigg( \sum_{k=1}^{\frac{n}{2}} Z_{k} - \!\!\!\!\!\!\!\! \sum_{k=1}^{\frac{n}{2} - \frac{n}{\log^{3}(n)}} \!\!\!\!\!\!\!\! W_{k} - cy_{i} \geq T + T\frac{\log(n)}{\log\log(n)}\bigg) \nonumber \\ 
\geq & \mathbb{P}\bigg( \sum_{k=1}^{\frac{n}{2}} [Z_{k} - W_{k}] \geq cy_{i} + T + T\frac{\log(n)}{\log\log(n)}\bigg) \nonumber \\ 
= & (1-\alpha) \mathbb{P}\bigg(  \frac{1}{l} \sum_{k=1}^{\frac{n}{2}} [Z_{k} - W_{k}] \geq \frac{1}{l} (c + T + T\frac{\log(n)}{\log\log(n)}) \bigg) \nonumber \\ 
& + \alpha \mathbb{P}\bigg(  \frac{1}{l} \sum_{k=1}^{\frac{n}{2}} [Z_{k} - W_{k}] \geq \frac{1}{l} (-c + T + T\frac{\log(n)}{\log\log(n)}) \bigg) \nonumber \\ 
\overset{(a)}{\geq} & (1-\alpha) e^{-l\big(t_{1}^{*}a_{1} - \Gamma(t_{1}^{*}) + |t_{1}^{*}|\delta\big) } (1-o(1))  \twocolbreak + \alpha e^{-l\big(t_{2}^{*}a_{2} - \Gamma(t_{2}^{*}) + |t_{2}^{*}|\delta\big) } (1-o(1))\label{eq.9}
\end{align}
where $(a)$ uses Lemma~\ref{Le.10} in Appendix~\ref{sec:app1-3} and the following definitions: $\delta \triangleq \frac{\log^{\frac{2}{3}}(n)}{l}$, $a_{1} \triangleq  \frac{1}{l} (c + T + T\frac{\log(n)}{\log\log(n)}) + \delta$,  $a_{2} \triangleq \frac{1}{l} (-c + T + T\frac{\log(n)}{\log\log(n)}) + \delta$, and:
\[
t_{1}^{*} = \arg\sup (ta_{1} - \Gamma(t)), \qquad t_{2}^{*} = \arg\sup ( ta_{2} - \Gamma(t))
\]
The supremum at $t_1^*$ is calculated as follows; $t_2^*$ is obtained similarly.
\begin{equation}
 ta_{1} - \Gamma(t) = ta_{1} - \log\big(1 - q(1-(\frac{a}{b})^{t}) \big) - \log\big(1 - p(1-(\frac{a}{b})^{-t}) \big)
\label{main.sup}
\end{equation}
The right hand side is concave in $t$, so we set the derivative to zero:
\begin{align}\label{main.diff}
&a_{1} - \frac{Tq(\frac{a}{b})^{t}}{1-q(1-(\frac{a}{b})^{t})} + \frac{Tp(\frac{a}{b})^{-t}}{1-p(1-(\frac{a}{b})^{-t})} \nonumber \\ 
& = \frac{\log(n)}{n} \bigg(\frac{2c}{\log(n)} + \frac{2T}{\log(n)} + \frac{2T}{\log\log(n)} + \frac{2}{\log^{\frac{1}{3}}(n)} - \twocolbreak \frac{Tb(\frac{a}{b})^{t}}{1-q(1-(\frac{a}{b})^{t})} + \frac{Ta(\frac{a}{b})^{-t}}{1-p(1-(\frac{a}{b})^{-t})}\bigg) \nonumber\\
&= 0
\end{align}

We consider two asymptotic regimes for $\alpha$:
\begin{itemize}
\item $c = o(\log(n))$. Then, the first four terms on the right hand side of~\eqref{main.diff} are $o(1)$. This suggests that $t^{*} = \frac{1}{2}$. Hence, substituting back in~\eqref{main.sup} leads to: 
\begin{align}
\twocolAlignMarker ta_{1} - \Gamma(t) \twocolnewline
&= \frac{1}{2} a_{1} \! - \! \log\big(1 \!-\! q(1\!-\! (\sqrt{\frac{a}{b}})) \big) \! - \! \log\big(1 \! - \! p(1\!-\! (\sqrt{\frac{b}{a}}) \big) \nonumber\\ 
& \overset{(a)}{\leq} \frac{1}{2} a_{1} + \frac{q(1-(\sqrt{\frac{a}{b}}))}{1 - q(1-(\sqrt{\frac{a}{b}}))} + \frac{p(1-(\sqrt{\frac{b}{a}}))}{1 - p(1-(\sqrt{\frac{b}{a}}))}\nonumber \\ 
& \overset{(b)}{=} \frac{\log(n)}{n} \bigg( (\sqrt{a} - \sqrt{b})^{2} + o(1) \bigg)\label{main.sup1}
\end{align}
where $(a)$ holds because $\log(1-x) \geq \frac{-x}{1-x}$ and $(b)$ holds because both $(1 - q(1-(\sqrt{\frac{a}{b}})))$ and $(1 - q(1-(\sqrt{\frac{a}{b}})))$ $\to 1$ as $n \to \infty$. Thus, we can bound the term involving $t_1^*$ as follows:
\begin{align}\nonumber
 & e^{-l\big(t_{1}^{*}a_{1} - \Gamma(t_{1}^{*}) + |t_{1}^{*}|\delta\big) } \geq e^{-\log(n)\big(\frac{1}{2} (\sqrt{a} - \sqrt{b})^{2} + o(1)\big)}
 \end{align}
We can similarly bound the term involving $t_2^*$ and substitute both in~\eqref{eq.9} to get:
\begin{align}\nonumber
\mathbb{P}(F_{i}^{H}) &\geq n^{-0.5(\sqrt{a}-\sqrt{b})^2 + o(1)} 
\end{align}
Thus, if $(\sqrt{a}-\sqrt{b})^2 \leq 2 - \varepsilon$ for some $0 < \varepsilon < 2$, then $\mathbb{P}(F_{i}^{H}) \geq n^{-1+\frac{\varepsilon}{2}} > \frac{\log^{3}(n)}{n} \log(\frac{1}{\delta})$ for $\delta\in (0,1)$ for sufficiently large $n$. This proves the first case of Lemma~\ref{Le.7}.

\item $c = \beta\log(n) + o(\log(n))$, $\beta >0$. Substituting in~\eqref{main.diff}, this suggests that $t_{1}^{*} = \frac{1}{T}\log(\frac{\gamma + \beta}{bT})$ and $t_{2}^{*} = \frac{1}{T}\log(\frac{\gamma - \beta}{bT})$, where $\gamma = \sqrt{\beta^{2} + abT^{2}}$. Hence, by substituting back in~\eqref{main.sup} and following the same ideas as in~\eqref{main.sup1}:
\begin{align}
\twocolAlignMarker ta_{1} - \Gamma(t) \twocolnewline
& \leq \frac{\log(n)}{n} \big ( 2\beta t^{*} \! + \! b(1\!-\! (\frac{a}{b})^{t^{*}})  + a(1\! -\! (\frac{a}{b})^{-t^{*}}) + o(1) \big) \nonumber\\
& = \frac{\log(n)}{n} \big( a+b+\beta -\frac{2\gamma}{T} + \frac{\beta}{T}\log(\frac{\gamma+\beta}{\gamma-\beta}) + o(1) \big) \nonumber \\ 
& = \frac{\log(n)}{n} (\eta(a,b,\beta) + o(1))\label{main.sup2}
\end{align}
We can then bound the term involving $t_1^*$ as follows:
\begin{align}\nonumber
& e^{-l\big(t_{1}^{*}a_{1} - \Gamma(t_{1}^{*}) + |t_{1}^{*}|\delta\big) } \geq e^{-\frac{\log(n)}{2}\big(\eta(a,b,\beta) + o(1)\big)}
\end{align}
We can similarly bound the term involving $t_2^*$ and substitute both in~\eqref{eq.9} to get:
\begin{align}\nonumber
\mathbb{P}(F_{i}^{H})& \geq  n^{-0.5\eta(a,b,\beta) + o(1)} + \alpha n^{-0.5\eta(a,b,\beta) + \beta + o(1)}\\ \nonumber
& = n^{-0.5\eta(a,b,\beta) + o(1)}
\end{align}
Thus, if $\eta(a,b,\beta) \leq 2 - \varepsilon$ for some $0 < \varepsilon < 2$, then $\mathbb{P}(F_{i}^{H}) \geq n^{-1+\frac{\epsilon}{2}} > \frac{\log^{3}(n)}{n} \log(\frac{1}{\delta})$ for $\delta\in (0,1)$ for sufficiently large $n$. This proves the second case of Lemma~\ref{Le.7}.

For the last case of Lemma~\ref{Le.7}, we begin as in~\eqref{eq.9} but take a different approach:
\begin{align}
&\mathbb{P}(F_{i}^{H}) \geq  \mathbb{P}\bigg( \sum_{k=1}^{\frac{n}{2}} [Z_{k} \!-\! W_{k}] \geq cy_{i} \!+\! T \!+\! T\frac{\log(n)}{\log\log(n)}\bigg) \nonumber \\ 
= & (1\!-\! \alpha) \Bigg(1\!-\! \mathbb{P}\bigg(\sum_{k=1}^{\frac{n}{2}} [Z_{k} \!-\! W_{k}] \leq c \!+\! T +\! T\frac{\log(n)}{\log\log(n)} \bigg)\Bigg) \nonumber\\
& + \alpha \Bigg(1\!-\! \mathbb{P}\bigg( \sum_{k=1}^{\frac{n}{2}} [Z_{k}\! - \!W_{k}] \leq -c \!+\! T \! + \! T\frac{\log(n)}{\log\log(n)} \bigg)\Bigg) \nonumber\\ 
\overset{(a)}{\geq} & 1- (1-\alpha)  \twocolbreak \times e^{-n \sup_{t>0} \frac{-t}{n}\big(c + T + T\frac{\log(n)}{\log\log(n)}\big) - \frac{1}{2} \log\big(\mathbb{E}(e^{-t[Z-W]})\big) } \nonumber \\
& \!\!-\alpha e^{-n \sup_{s>0} \frac{-s}{n}\big(-c + T + T\frac{\log(n)}{\log\log(n)}\big) - \frac{1}{2} \log\big(\mathbb{E}(e^{-s[Z- W]})\big) } \label{eq.9_2}
\end{align}
where $(a)$ is a Chernoff bound. 
A direct computation of the logarithmic term leads to:
\begin{align}
\twocolAlignMarker \log\big( \mathbb{E}\big[e^{-t[Z - W]}\big]\big) \twocolnewline 
&\overset{(a)}{=}  \log\big(1 \!- \!q(1-\big(\frac{p}{q}\big)^{-t}) \!+\! \log\big(1 \!-\! p(1-\big(\frac{p}{q}\big)^{t})\big) \nonumber\\ 
&  \overset{(b)}{\leq}-q(1-\big(\frac{p}{q}\big)^{-t}) - p(1-\big(\frac{p}{q}\big)^{t})\label{eq.9_2_2}
\end{align}
where $(a)$ follows from the fact that $W_{i}$, $Z_{i}$ are independent random variables $\forall i$, and $(b)$ holds because $\log(1-x) \leq -x$. Substituting~\eqref{eq.9_2_2} into~\eqref{eq.9_2} yields:
\begin{align}
&\mathbb{P}(F_{i}^{H}) \geq 1 \twocolbreak -(1-\alpha)   e^{- \frac{\log(n)}{n} \sup_{t>0} -t(\beta+o(1)) + \frac{1}{2} \big( a+b-a(\frac{a}{b})^{t} - b(\frac{a}{b})^{-t}\big) } \nonumber \\
& -\alpha e^{- \frac{\log(n)}{n} \sup_{s>0} [s(\beta+o(1)) + \frac{1}{2} ( a+b-a(\frac{a}{b})^{s} - b(\frac{a}{b})^{-s}) ]}  \label{eq.9_2_2_2}
\end{align}
Recall that $\beta>0$. Since  $-t(\beta+o(1)) + \frac{1}{2} \big( a+b-a(\frac{a}{b})^{t} - b(\frac{a}{b})^{-t}\big)$ is concave in $t$. We find its equilibrium by taking the derivative:
\begin{align}\label{main.diff2}
-\beta - \frac{aT}{2}(\frac{a}{b})^{t} +  \frac{bT}{2}(\frac{a}{b})^{-t} = 0
\end{align}
The derivative has a zero at $\frac{1}{T}\log(\frac{\gamma - \beta}{aT})$ which is negative due to positivity of $\beta$, therefore by continuity, the supremum over $t>0$ is achieved at $t^*=0$. Similarly the supremum over $s$ can be calculated via a derivative, finding $s^*=(\frac{1}{T}\log(\frac{\gamma + \beta}{aT}))^+$, which is positive as long as $\beta > \frac{T(a-b)}{2}$. Since $s^*=0$ leads to a trivial bound, consider $\beta>\frac{T(a-b)}{2}$ and substitute in~\eqref{eq.9_2}.
\begin{align}\nonumber
\mathbb{P}(F_{i}^{H}) & \geq 1 - (1-\alpha) e^{0} - \alpha n^{-\frac{1}{2}\eta(a,b,\beta) + \beta} \\ \nonumber
& = n^{-\beta} - n^{-\frac{1}{2}\eta(a,b,\beta)}
\end{align}
using $\alpha = n^{-\beta}$. Hence, if $\beta \leq 1 - \varepsilon_{1}$ and $\frac{1}{2}\eta \geq 1+\varepsilon_{2}$, then $\mathbb{P}(F_{i}^{H}) \geq n^{-1} (n^{\varepsilon_{1}} - n^{-\varepsilon_{2}}) > \frac{\log^{3}(n)}{n} \log(\frac{1}{\delta})$ for $\delta\in (0,1)$ for sufficiently large $n$. This proves the third and last case of Lemma~\ref{Le.7}.

\end{itemize}



\section{Proof of Lemma~\ref{Le.13}}
\label{app.3}
By upper bounding $P_{e}$, we get:
\begin{align}
P_{e} \leq & \mathbb{P}\bigg(\sum_{k=1}^{(1-\delta)\frac{n}{2}} Z_{k} + \sum_{k=1}^{\delta\frac{n}{2}} W_{k} \geq \sum_{j=1}^{(1-\delta)\frac{n}{2}-\frac{2D}{\log(n)}n} W_{j}  + \frac{c}{T}y_{i}\bigg) \nonumber\\ 
\leq & \mathbb{P}\bigg(\sum_{k=1}^{\frac{n}{2}} Z_{k} + \sum_{k=1}^{\delta\frac{n}{2}} W_{k} \geq \sum_{j=1}^{(1-\delta)\frac{n}{2}-\frac{2D}{\log(n)}n} W_{j}  + \frac{c}{T}y_{i}\bigg) \nonumber\\ 
\leq & \mathbb{P}\bigg(\sum_{k=1}^{\frac{n}{2}} Z_{k} - \sum_{k=1}^{\frac{n}{2}} W_{k} + \sum_{j=1}^{\delta n + \frac{2D}{\log(n)}n} W_{j}  \geq  \frac{c}{T}y_{i}\bigg) \nonumber
\end{align}
Defining  $\psi \triangleq \frac{1}{\sqrt{-\log(\delta)}}$
\begin{align}
P_e\leq & \mathbb{P}\bigg(\sum_{k=1}^{\frac{n}{2}} (Z_{k} - W_{k}) \geq \frac{c}{T}y_{i} - \psi\log(n)\bigg) + \twocolbreak
\mathbb{P}\bigg(\sum_{j=1}^{\delta n + \frac{2D}{\log(n)}n} W_{j} \geq  \psi\log(n)\bigg) \nonumber\\ 
= & (1-\alpha)\mathbb{P}\bigg(\sum_{k=1}^{\frac{n}{2}} (Z_{k} - W_{k}) \geq \frac{c}{T} - \psi\log(n)\bigg) + \nonumber\\ 
&\alpha \mathbb{P}\bigg(\sum_{k=1}^{\frac{n}{2}} (Z_{k} - W_{k}) \geq -\frac{c}{T} - \psi\log(n)\bigg) +  \twocolbreak \mathbb{P}\bigg(\sum_{j=1}^{\delta n + \frac{2D}{\log(n)}n} W_{j} \geq  \psi\log(n)\bigg)
\label{eq.ap.9}
\end{align}
Similar to Lemma~\ref{Le.5}, we use a multiplicative Chernoff bound on the sum of i.i.d.\ random variables $W_j$:
\[
\mathbb{P}\big(\sum_{j=1}^{m} W_j \geq t\mu\big) \leq \big(\frac{t}{e}\big)^{-t\mu} \quad,\quad \mu = m\, \mathbb{E}[W_j].
\]
where $m=\delta n + \frac{2D}{\log(n)}n$,  $\mu = a(\delta\log(n) + 2D)$ and we set $t = \frac{\psi\log(n)}{a(\delta\log(n) + 2D)}$ to get:
\begin{align}
&\mathbb{P}\bigg(\!\!\!\!\sum_{j=1}^{\delta n + \frac{2D}{\log(n)}n} \!\!\!\!\!\!\! W_{j} \geq  \psi\log(n)\bigg) \leq \bigg(\frac{\psi\log(n)}{ae(\delta\log(n) + 2D)}\bigg)^{-\psi\log(n)} \nonumber\\ 
& = \bigg(\frac{\psi}{a\delta e(1 + \frac{2D}{\delta\log(n)} )}\bigg)^{-\frac{\log(n)}{\sqrt{\log(\frac{1}{\delta})}}} \nonumber\\ 
& = e^{\log(n) \bigg( \frac{1+\log(a)}{\sqrt{\log(\frac{1}{\delta})}} + \frac{\log(1 + \frac{2D}{\delta\log(n)})}{\sqrt{\log(\frac{1}{\delta})}} + \frac{\log(\delta) + \frac{1}{2}\log\log(\frac{1}{\delta})}{\sqrt{\log(\frac{1}{\delta})}} \bigg)} \nonumber\\ 
& = n^{-\sqrt{\log(\frac{1}{\delta})} \bigg( 1 - \frac{\log(1 + \frac{2D}{\delta\log(n)})}{\log(\frac{1}{\delta})}  + o(1)\bigg)  }\label{eq.ap.10}
\end{align}
where we used $\lim_{D \to \infty}\delta=0$. Since there exists $D$ sufficiently large such that $\frac{\log(1 + \frac{2D}{\delta\log(n)})}{\log(\frac{1}{\delta})} < 1$, 
\begin{equation}\label{eq.ap.11}
\mathbb{P}\bigg(\sum_{j=1}^{\delta n + \frac{2D}{\log(n)}n} W_{j} \geq  \psi\log(n)\bigg) \leq n^{-(1+\Omega(1))}\quad .
\end{equation}

Chernoff bound can be applied to the first term in~\eqref{eq.ap.9}:
\begin{align}\label{eq.ap.12}
&  (1-\alpha)\mathbb{P}\bigg(\sum_{k=1}^{\frac{n}{2}} (Z_{k} - W_{k}) \geq \frac{c}{T} - \psi\log(n)\bigg) + \twocolbreak \alpha \mathbb{P}\bigg(\sum_{k=1}^{\frac{n}{2}} (Z_{k} - W_{k}) \geq -\frac{c}{T} - \psi\log(n)\bigg) \nonumber \\ 
& \overset{(a)}{\leq} (1-\alpha) e^{-\frac{\log(n)}{2} \sup_{t_{1}>0} 2t_1 (\frac{c}{T\log(n)} - \psi) + a + b - be^{t_{1}} - a e^{-t_{1}}} + \twocolbreak
  \alpha e^{-\frac{\log(n)}{2} \sup_{t_{2}>0} 2t_{2} (-\frac{c}{T\log(n)} - \psi) + a + b - be^{t_{2}} - a e^{-t_{2}}}
\end{align}
where $(a)$ holds because $\log(1-x) \leq -x $. Since $\psi \to 0$ as $D \to \infty$, $\psi$ can be replaced by $o(1)$ for sufficiently large $D$. We consider the following asymptotic regimes for $\alpha$.
\begin{itemize}
\item If $c = o(\log(n))$, this suggests that $t_{1}^{*} = t_{2}^{*} = \frac{1}{2}T$. Hence, ~\eqref{eq.ap.12} can be upper bounded by: 
    \begin{align}\label{eq.ap.13}
    & n^{-\frac{1}{2}(\sqrt{a}-\sqrt{b})^{2} + o(1)}
    \end{align}

\item If $c = \beta\log(n)+o(\log(n))$, for $ 0 < \beta < \frac{T(a-b)}{2}$, then it can be shown that $t_{1}^{*} = \log(\frac{\gamma + \beta}{bT})$ and $t_{2}^{*} = \log(\frac{\gamma - \beta}{bT})$, where $\gamma = \sqrt{\beta^{2} + abT^{2}}$. Hence, ~\eqref{eq.ap.12} can be upper bounded by: 
    \begin{align}\label{eq.ap.14} 
    & (2-\alpha) n^{-\frac{1}{2}\eta(a,b,\beta) + o(1)}
    \end{align}

\item If $c = \beta\log(n)$, for $ \beta > \frac{T(a-b)}{2}$, then it can be shown that $t_{1}^{*} = \log(\frac{\gamma + \beta}{bT})$ and $t_{2}^{*} = 0$. Hence, ~\eqref{eq.ap.12} can be upper bounded by: 
    \begin{align}\label{eq.ap.15} 
    & (1-\alpha) n^{-\frac{1}{2}\eta(a,b,\beta) + o(1)} + n^{-\beta}
    \end{align}

\end{itemize}
The last three equations and~\eqref{eq.ap.11}, substituting in~\eqref{eq.ap.9}, concludes the proof of the lemma.


\section{Proof of Lemma~\ref{Le.8}}
\label{sec:app1-2}

Define $l = \frac{n}{2}$ and let $\Gamma(t) \triangleq \log(\mathbb{E}_{X}[e^{tx}])$ for a random variable $X$. Then,
\begin{align}
\mathbb{P}(F_{i}^{H}) &  \nonumber = \epsilon \mathbb{P}\bigg( \sum_{k=1}^{\frac{n}{2}} (Z_{k}) - \sum_{k=1}^{\frac{n}{2} - \frac{n}{\log^{3}(n)}} (W_{k}) \geq 1 + \frac{\log(n)}{\log\log(n)}\bigg) \\ \nonumber
& \geq \epsilon \mathbb{P}\bigg( \sum_{k=1}^{\frac{n}{2}} [Z_{k} - W_{k}] \geq  1 + \frac{\log(n)}{\log\log(n)}\bigg) \\ \nonumber
& \overset{(a)}{\geq} \epsilon e^{-l\big(t^{*}a -\Gamma(t_{1}^{*}) + |t_{1}^{*}|\delta\big) } (1-o(1)) \\ 
& =  e^{-l\big(t^{*}a -\Gamma(t_{1}^{*}) + |t_{1}^{*}|\delta\big) + \log(\epsilon) } (1-o(1))\label{eq.s2.9}
\end{align}
where $(a)$ holds by defining $\delta = \frac{\log^{\frac{2}{3}}(n)}{l}$, $a = \frac{1}{l} (1 + \frac{\log(n)}{\log\log(n)}) + \delta$, $t_{1}^{*} = \arg\sup_{t \in \mathbb{R}} at - \Gamma(t)$ and by using Lemma~\ref{Le.10} in Appendix~\ref{sec:app1-3}.

In a manner similar to~\eqref{main.sup} and~\eqref{main.diff}, it can be shown that $t^{*} = \frac{1}{2}T$. Substituting in~\eqref{eq.s2.9} and using $\log(1-x) \geq \frac{-x}{1-x}$:
\begin{equation}
\mathbb{P}(F_{i}^{H}) \geq \epsilon n^{-0.5(\sqrt{a}-\sqrt{b})^2 + o(1)} 
\end{equation}

When $\log(\epsilon) = o(\log(n))$ and $(\sqrt{a}-\sqrt{b})^2 \leq 2 - \varepsilon$ for some $0 < \varepsilon < 2$, for sufficiently large $n$ and all $\delta\in (0,1)$ we have:
\[
\mathbb{P}(F_{i}^{H}) \geq n^{-1+\frac{\varepsilon}{2}} > \frac{\log^{3}(n)}{n} \log(\frac{1}{\delta}).
\]
This proves the first case of Lemma~\ref{Le.8}. 

When $\log(\epsilon) = -\beta\log(n)+o(\log(n))$ for positive $\beta$, and $(\sqrt{a}-\sqrt{b})^2 + 2\beta \leq 2 - \varepsilon$ for some $0 < \varepsilon < 2$, then for sufficiently large $n$ and all $\delta\in (0,1)$ we have:
\[
\mathbb{P}(F_{i}^{H}) \geq n^{-1+\frac{\varepsilon}{2}} > \frac{\log^{3}(n)}{n} \log(\frac{1}{\delta})
\]
This proves the second and last case of Lemma~\ref{Le.8}.


\section{Proof of Lemma~\ref{Le.9}}
\label{sec:app1-5}

Let $W_{i} \sim Bern(p)$, $Z_{i} \sim Bern(q)$ and define $l = \frac{n}{2}$ and $\Gamma(t) \triangleq \log(\mathbb{E}_{X}[e^{tx}])$ for a random variable $X$. Then, we have the following:
\begin{align}
\twocolAlignMarker \mathbb{P}(F_{i}^{H_{1}}) \twocolnewline
&=  \mathbb{P}\bigg( \sum_{j=1}^{\frac{n}{2}} (Z_{j}) - \sum_{j=1}^{\frac{n}{2} - \frac{n}{\log^{3}(n)}} (W_{j}) \geq \frac{\hbar_i}{T} + 1 + \frac{\log(n)}{\log\log(n)}\bigg)  \nonumber \\ 
&\geq  \sum_{m_{1}=1}^{M_{1}}\sum_{m_{2}=1}^{M_{2}}\cdots\sum_{m_{K}=1}^{M_{K}} (\prod_{k=1}^{K}\alpha_{+,m_{k}}^{k}) \twocolbreak \mathbb{P}\bigg( \sum_{j=1}^{\frac{n}{2}} [Z_{j} - W_{j}] \geq \sum_{k=1}^{K}\frac{\hfeat}{T} + 1 + \frac{\log(n)}{\log\log(n)}\bigg) \nonumber \\
& \overset{(a)}{\geq} \sum_{m_{1}=1}^{M_{1}}\sum_{m_{2}=1}^{M_{2}}\cdots\sum_{m_{K}=1}^{M_{K}} (\prod_{k=1}^{K}\alpha_{+,m_{k}}^{k}) e^{-l\big(t^{*}a -\Gamma(t^{*}) + |t^{*}|\delta\big) }  \label{eq.s3.9}
\end{align}
where $(a)$ holds by defining $\delta = \frac{\log^{\frac{2}{3}}(n)}{l}$, $a = \frac{1}{l} (\sum_{k=1}^{K}\frac{\hfeat}{T} + 1 + \frac{\log(n)}{\log\log(n)}) + \delta$, $t^{*} = \arg\sup_{t \in \mathbb{R}} a t - \Gamma(t)$ and by using Lemma~\ref{Le.10}. For convenience $a$ and $t^{*}$ have no subscripts even though both depend on feature outcomes. 
Similarly,
\begin{align}
\mathbb{P}(F_{j}^{H_{2}}) &   
\geq \sum_{m_{1}=1}^{M_{1}}\sum_{m_{2}=1}^{M_{2}}\cdots\sum_{m_{K}=1}^{M_{K}} (\prod_{k=1}^{K}\alpha_{+,m_{k}}) e^{-l\big(t^{*}a -\Gamma(t^{*}) + |t^{*}|\delta\big) }   \label{eq.s3.9.1}
\end{align}
where $a = \frac{1}{l} (-\sum_{k=1}^{K}\frac{\hfeat}{T} + 1 + \frac{\log(n)}{\log\log(n)}) + \delta$.

Without loss of generality, we focus on one term of the nested sum in~\eqref{eq.s3.9} and~\eqref{eq.s3.9.1}. Then,
\begin{itemize}
\item If $\sum_{k=1}^{K} \hfeat = o(\log(n))$ and both $\sum_{k=1}^{K} \log(\alpha_{+,m_{k}}^{k})$ and $\sum_{k=1}^{K} \log(\alpha_{-,m_{k}}^{k})$ are $ o(\log(n))$, then the optimal $t$ for that term is $t^{*} = \frac{1}{2}T$ for both~\eqref{eq.s3.9},~\eqref{eq.s3.9.1}. Hence, substituting in~\eqref{eq.s3.9}, ~\eqref{eq.s3.9.1} leads to:
\begin{align}
\mathbb{P}(F_{i}^{H_{1}}) &\geq  n^{-0.5(\sqrt{a}-\sqrt{b})^2 + o(1)} \\
\mathbb{P}(F_{j}^{H_{2}}) &\geq  n^{-0.5(\sqrt{a}-\sqrt{b})^2 + o(1)}
\end{align}
Thus, if $(\sqrt{a}-\sqrt{b})^2 \leq 2 - \varepsilon$ for some $0 < \varepsilon < 2$, then $\mathbb{P}(F_{i}^{H_{1}})$ and $\mathbb{P}(F_{j}^{H_{2}})$ are both greater than $n^{-1+\frac{\varepsilon}{2}} > \frac{\log^{3}(n)}{n} \log(\frac{1}{\delta})$ for $\delta \in (0,1)$ for sufficiently large $n$.

\item If $\sum_{k=1}^{K} \hfeat =  o(\log(n)), \sum_{k=1}^{K}\log(\alpha_{+,m_{k}}^{k}) = \sum_{k=1}^{K} \log(\alpha_{-,m_{k}}^{k})=-\beta\log(n)+ o(\log(n)), \beta > 0$, then $t^{*} = \frac{1}{2}T$ for both~\eqref{eq.s3.9},~\eqref{eq.s3.9.1}. Hence, by substituting in~\eqref{eq.s3.9}, ~\eqref{eq.s3.9.1}:
\begin{align}
\mathbb{P}(F_{i}^{H_{1}}) &\geq  n^{-0.5(\sqrt{a}-\sqrt{b})^2 - \beta + o(1)}\\
\mathbb{P}(F_{j}^{H_{2}}) &\geq  n^{-0.5(\sqrt{a}-\sqrt{b})^2 - \beta + o(1)}
\end{align}
Thus, if $(\sqrt{a}-\sqrt{b})^2 + 2\beta \leq 2 - \varepsilon$ for some $0 < \varepsilon < 2$, then $\mathbb{P}(F_{i}^{H_{1}})$ and $\mathbb{P}(F_{j}^{H_{2}})$ are both greater than $n^{-1+\frac{\varepsilon}{2}} > \frac{\log^{3}(n)}{n} \log(\frac{1}{\delta})$ for $\delta \in (0,1)$ for sufficiently large $n$.

 \item If $\sum_{k=1}^{K} \hfeat = \beta\log(n)+ o(\log(n)),  0<\beta < T\frac{(a-b)}{2}$, then $t^{*} = \log(\frac{\gamma + \beta}{bT})$ for~\eqref{eq.s3.9} and $t^{*} = \frac{1}{T}\log(\frac{\gamma - \beta}{bT})$ for~\eqref{eq.s3.9.1}. Hence, by substituting in~\eqref{eq.s3.9}, ~\eqref{eq.s3.9.1}:
\begin{align*}
\mathbb{P}(F_{i}^{H_{1}}) &\geq  e^{-\log(n)\big( 0.5\eta(a,b,\beta) - \sum_{k=1}^{K}\frac{\log(\alpha_{+,m_{k}}^{k})}{\log(n)} + o(1) \big) } \\
\mathbb{P}(F_{j}^{H_{2}}) &\geq  e^{-\log(n)\big( 0.5\eta(a,b,\beta) - \beta - \sum_{k=1}^{K}\frac{\log(\alpha_{-,m_{k}}^{k})}{\log(n)} + o(1) \big) } 
\end{align*}
Then, if $\sum_{k=1}^{K} \log(\alpha_{+,m_{k}}^{k}) = o(\log(n))$, this implies that $\sum_{k=1}^{K} \frac{\log(\alpha_{-,m_{k}}^{k})}{\log(n)} = -\beta + o(1)$. Hence,
\begin{equation}
\mathbb{P}(F_{i}^{H_{1}}) \geq  n^{- 0.5\eta(a,b,\beta) + o(1)} 
\end{equation}
\begin{equation}
\mathbb{P}(F_{j}^{H_{2}}) \geq  n^{- 0.5\eta(a,b,\beta) + o(1)}
\end{equation}
Thus, if $\eta(a,b,\beta) \leq 2 - \varepsilon$ for some $0 < \varepsilon < 2$, then $\mathbb{P}(F_{i}^{H_{1}})$ and $\mathbb{P}(F_{j}^{H_{2}})$ are both greater than $n^{-1+\frac{\varepsilon}{2}} > \frac{\log^{3}(n)}{n} \log(\frac{1}{\delta})$ for $\delta \in (0,1)$ for sufficiently large $n$.

If $\sum_{k=1}^{K} \log(\alpha_{+,m_{k}}^{k}) = -\beta^{'}\log(n) + o(\log(n))$, this implies that $\sum_{k=1}^{K}\frac{\log(\alpha_{-,m_{k}}^{k})}{\log(n)} = -\beta^{''}$, for some $\beta^{''} > 0$ and $\beta = \beta^{''} - \beta^{'}$. Hence,
\begin{align}
\mathbb{P}(F_{i}^{H_{1}}) &\geq  n^{- 0.5\eta(a,b,\beta) - \beta^{'} + o(1)} \\
\mathbb{P}(F_{j}^{H_{2}}) & \geq  n^{- 0.5\eta(a,b,\beta) + \beta - \beta^{''} + o(1)} \nonumber\\ 
&=  n^{- 0.5\eta(a,b,\beta) - \beta^{'} + o(1)} 
\end{align}
Thus, it is clear that if $\eta(a,b,\beta) + 2\beta^{'} \leq 2 - \varepsilon$ for some $0 < \varepsilon < 2$, then $\mathbb{P}(F_{i}^{H_{1}})$ and $\mathbb{P}(F_{j}^{H_{2}})$ are both greater than $n^{-1+\frac{\varepsilon}{2}} > \frac{\log^{3}(n)}{n} \log(\frac{1}{\delta})$ $\delta \in (0,1)$ for sufficiently large $n$. The case when $-T\frac{(a-b)}{2}<\beta < 0$ holds similarly.
\end{itemize}

\section{Proof of Lemma~\ref{Le.10}}
\label{sec:app1-3}
\begin{lemma}\label{Le.10}
Let $X_{1},\cdots,X_{n}$ be a sequence of i.i.d random variables. Define $\Gamma(t) = \log(\mathbb{E}[e^{tX}])$. Then, for any $a ,\epsilon \in \mathbb{R}$:
\begin{align}\nonumber
&\mathbb{P}\big(\frac{1}{n} \sum_{i=1}^{n} X_{i} \geq a-\epsilon\big) \geq e^{-n\big(t^{*}a-\Gamma(t^{*})+|t^{*}|\epsilon\big)} \bigg( 1 - \frac{\sigma^{2}_{\hat{X}}}{n\epsilon^{2}} \bigg)
\end{align}
where $t^{*} = \arg\sup_{t\in \mathbb{R}} (ta - \Gamma(t))$, $\hat{X}$ is a random variable with the same alphabet as $X$ but distributed according to $\frac{e^{t^{*}x}\mathbb{P}(x)}{\mathbb{E}_{X}[e^{t^{*}X}]}$ and $\mu_{\hat{X}}, \sigma^{2}_{\hat{X}}$ are the mean and variance of $\hat{X}$, respectively.
\end{lemma}

\begin{proof}
\begin{align}\nonumber
&\mathbb{P}\big(\frac{1}{n} \sum_{i=1}^{n} X_{i} \geq a-\epsilon\big) \geq \mathbb{P}\big(a-\epsilon \leq \frac{1}{n} \sum_{i=1}^{n} X_{i} \leq a+\epsilon\big) \\ \nonumber
= &  \int_{|\frac{1}{n} \sum x_{i}-a| \leq \epsilon}  \mathbb{P}(x_{1}) \cdots \mathbb{P}(x_{n}) dx_{1}\cdots dx_{n}\\ \nonumber
\overset{(a)}{\geq} &  e^{-n(ta-\Gamma(t)+|t|\epsilon)}  \int_{| \frac{1}{n} \sum x_{i} -a|\leq \epsilon} \, \prod_{i=1}^{n} \bigg( \frac{e^{tx_{i}}\mathbb{P}(x_{i})}{\mathbb{E}_{X}[e^{tX}]} dx_{i}\bigg)\\ \nonumber
\overset{(b)}{=} &  e^{-n(ta-\Gamma(t)+|t|\epsilon)} \, \mathbb{P}\bigg( a-\epsilon \leq \frac{1}{n} \sum_{i=1}^{n} \hat{X}_{i} \leq a+\epsilon \bigg) \\
\overset{(c)}{\geq} &  e^{-n(ta-\Gamma(t)+|t|\epsilon)} \bigg( 1 - \frac{ n\sigma^{2}_{\hat{X}} + (n\mu_{\hat{X}} - na)^{2} }{n^{2}\epsilon^{2}} \bigg)
\label{lem.ap}
\end{align}
where for all finite $\mathbb{E}[e^{tX}]$, $(a)$ is true becuase $e^{t\sum x_i} \le e^{n(ta+|t|\epsilon)}$ over the range of integration, $(b)$ holds because $\frac{e^{t x}\mathbb{P}_X(x)}{\mathbb{E}_{X}[e^{tX}]}$ is a valid distribution~\cite{large_dev}, and $(c)$ holds by Chebyshev inequality.

Since $(ta - \Gamma(t))$ is concave in $t$~\cite{large_dev}, to find $t^*=\arg \sup_t (ta - \Gamma(t))$ we set the derivative to zero, finding $a = \frac{\mathbb{E}_{X}[Xe^{t^{*}X}]}{\mathbb{E}[e^{t^{*}X}]}$. Also, by direct computation $\mu_{\hat{X}} = \frac{\mathbb{E}_{X}[Xe^{tX}]}{\mathbb{E}[e^{tX}]}$. This means that at $t = t^{*}$, we have $\mu_{\hat{X}} = a$. Thus, substituting back in~\eqref{lem.ap} leads to:
\begin{align*}
&\mathbb{P}\bigg(\frac{1}{n} {\sum}_{i=1}^{n} X_{i} \geq a-\epsilon\bigg) \geq e^{-n(t^{*}a-\Gamma(t^{*})+|t^{*}|\epsilon)} \big( 1 - \frac{\sigma^{2}_{\hat{X}}}{n\epsilon^{2}} \big)
\end{align*}
In our model $\epsilon = \frac{\log^{\frac{2}{3}}(n)}{n}$ and $X = T(Z - W)$, where $Z \sim \text{Bern}(q)$ and $W \sim \text{Bern}(p)$, where $T=\log(\frac{a}{b})$. Hence, $\sigma^{2}_{\hat{X}}= O(\frac{\log(n)}{n})$, therefore
\begin{align}\nonumber
&\mathbb{P}\bigg(\frac{1}{n} {\sum_{i=1}^{n}} X_{i} \geq a-\epsilon\bigg) \geq e^{-n(t^{*}a-\Gamma(t^{*})+|t^{*}|\epsilon)} \big( 1 - o(1) \big)
\end{align}
which concludes the proof.
\end{proof}
\bibliographystyle{IEEEtran}
\bibliography{IEEEabrv,ISIT2016_updated}

\end{document}